%% file: main.tex
\documentclass[11pt]{article}

\usepackage[T5]{fontenc}

\usepackage{amsmath,amssymb,amsthm,algpseudocode,float,caption,mathtools,thmtools,thm-restate,scalerel,stackengine,verbatim}
\usepackage{algorithm,algpseudocode}
\usepackage[shortlabels]{enumitem}

\stackMath
\newcommand\reallywidehat[1]{%
\savestack{\tmpbox}{\stretchto{%
  \scaleto{%
      \scalerel*[\widthof{\ensuremath{#1}}]{\kern-.6pt\bigwedge\kern-.6pt}%
          {\rule[-\textheight/2]{1ex}{\textheight}}%WIDTH-LIMITED BIG WEDGE
            }{\textheight}% 
            }{0.5ex}}%
            \stackon[1pt]{#1}{\tmpbox}%
            }
\usepackage[margin=1in]{geometry}
\usepackage{multirow}
\usepackage{subcaption}
\usepackage{afterpage}
\usepackage{bm}
\usepackage{tikz}
\usetikzlibrary{decorations.pathreplacing, decorations.markings, positioning}

\usepackage{bbm}
\usepackage[hidelinks]{hyperref}
\usepackage{nameref}
\usepackage{mathrsfs}

\newcommand{\wt}[1]{\widetilde{#1}}

\newcommand{\wh}[1]{\widehat{#1}}
\newcommand{\abs}[1]{\left|#1\right|}
\newcommand{\eps}{\varepsilon}

\newcommand{\Nbb}{\mathbb{N}}

\newcommand{\floor}[1]{\left\lfloor#1\right\rfloor}

\DeclarePairedDelimiter\brac{\lbrack}{\rbrack}
\DeclarePairedDelimiter\set{\lbrace}{\rbrace}
\DeclarePairedDelimiter\paren{\lparen}{\rparen}

\newcommand{\E}[2][]{\operatorname*{\mathbb{E}}_{#1 }\brac*{#2}}
\newcommand{\Pb}[2][]{\operatorname*{Pr}_{#1 }\brac*{#2}}

\newcommand{\bO}[1]{\operatorname*{O}\paren*{#1}}

\newcommand{\bOm}[1]{\operatorname*{\Omega}\paren*{#1}}

\newcommand{\bTh}[1]{\operatorname*{\Theta}\paren*{#1}}

\newcommand{\Xb}{\mathbf{X}}

\newcommand{\vb}{\mathbf{v}}

\newcommand{\Ac}{\mathcal{A}}

\newcommand{\Uc}{\mathcal{U}}

\newcommand{\Fc}{\mathcal{F}}

\newcommand{\Ic}{\mathcal{I}}

\DeclareMathOperator*{\ind}{{\normalfont\text{INDEX}}}

\newcommand{\e}{{\eps}}
\newcommand{\bool}{{\{0, 1\}}}
\newcommand{\R}{{\mathbb{R}}}

\newcommand{\expect}{{\mathbb E}}
\newcommand{\F}{{\mathcal{F}}}

\newcommand{\I}{{\mathcal I}}

\DeclareMathOperator{\poly}{poly}

\newtheorem{theorem}{Theorem}[section]
\newtheorem*{theorem*}{Theorem}
\newtheorem{lemma}[theorem]{Lemma}
\newtheorem{definition}[theorem]{Definition}
\newtheorem*{definition*}{Definition}
\newtheorem*{lemma*}{Lemma}
\newtheorem{corollary}[theorem]{Corollary}
\newtheorem*{corollary*}{Corollary}
\newtheorem{claim}[theorem]{Claim}
\newtheorem*{claim*}{Claim}

\newtheorem{remark}[theorem]{Remark}

\newcommand{\walksreset}{{\normalfont\textsc{SamplesWithReset}}}
\newcommand{\walks}{{\normalfont\textsc{SimulateWalks}}}

\newcommand{\walksresetH}{{\normalfont\textsc{SamplesWithResetHybrid}}}
\newcommand{\walksH}{{\normalfont\textsc{SimulateWalksHybrid}}}

\newcommand{\approxrp}{{\normalfont\textsc{ApproxRP}}}
\newcommand{\approxpagerank}{{\normalfont\textsc{ApproxPageRank}}}

\newcommand{\wft}{{\normalfont\textsc{WalkFromTemplate}}}
\newcommand{\wftH}{{{\normalfont\textsc{WalkFromTemplateHybrid}}}}

\newcommand{\vv}{\vec{\mathbf{v}}}
\newcommand{\ww}{\vec{\mathbf{w}}}

\newcommand{\nwalks}{b}

\declaretheorem[name=Lemma]{lem,numberwithin=section,sibling=thm}

\newenvironment{proofof}[1]{\noindent{\bf Proof of #1:}}{$\qed$\par}

{\makeatletter
	\gdef\xxxmark{%
		\expandafter\ifx\csname @mpargs\endcsname\relax % in minipage?
		\expandafter\ifx\csname @captype\endcsname\relax % in figure/caption?
		\marginpar{xxx}% not in a caption or minipage, can use marginpar
		\else
		xxx % notice trailing space
		\fi
		\else
		xxx % notice trailing space-
		\fi}
	\gdef\xxx{\@ifnextchar[\xxx@lab\xxx@nolab}
	\long\gdef\xxx@lab[#1]#2{{\bf [\xxxmark #2 ---{\sc #1}]}}
	\long\gdef\xxx@nolab#1{{\bf [\xxxmark #1]}}
}

\makeatletter
\let\orgdescriptionlabel\descriptionlabel
\renewcommand*{\descriptionlabel}[1]{%
  \let\orglabel\label
  \let\label\@gobble
  \phantomsection
 \edef\@currentlabel{#1\unskip}%
  \let\label\orglabel
  \orgdescriptionlabel{#1}%
}
\makeatother

\title{Simulating Random Walks in Random Streams}
\date{}
\author{John Kallaugher\\UT Austin \and Michael Kapralov\\EPFL \and Eric Price\\UT Austin}
\begin{document}

\maketitle
\begin{abstract}
The random order graph streaming model has received significant attention
recently, with problems such as matching size estimation, component counting,
and the evaluation of bounded degree constant query testable properties shown
to admit surprisingly space efficient algorithms. 

The main result of this paper is a space efficient single pass random order
streaming algorithm for simulating nearly independent random walks that start
at uniformly random vertices. We show that the distribution of $k$-step walks
from $\nwalks$ vertices chosen uniformly at random can be approximated up to error
$\varepsilon$ per walk  using $(1/\varepsilon)^{O(k)} 2^{O(k^2)}\cdot \nwalks$ words of space
with a single pass over a randomly ordered stream of edges, solving an open
problem of Peng and Sohler [SODA~`18]. Applications of our result include the
estimation of the average return probability of the $k$-step walk (the trace of
the $k^\text{th}$ power of the random walk matrix) as well as the estimation of
PageRank. We complement our algorithm with a strong impossibility result for
directed graphs.
\end{abstract}

\thispagestyle{empty}
\setcounter{page}{0}

\newpage

\tableofcontents
\setcounter{page}{0}
\thispagestyle{empty}
\newpage

\input{intro.tex}

\input{rp.tex}

\input{digraph-lb.tex}

\section*{Acknowledgements}
Michael Kapralov was supported by ERC Starting Grant 759471.

John Kallaugher and Eric Price were supported by NSF Award CCF-1751040
(CAREER). 

John was also supported by Laboratory Directed Research and Development program
at Sandia National Laboratories, a multimission laboratory managed and operated
by National Technology and Engineering Solutions of Sandia, LLC., a wholly
owned subsidiary of Honeywell International, Inc., for the U.S. Department of
Energy’s National Nuclear Security Administration under contract DE-NA-0003525.
Also supported by the U.S. Department of Energy, Office of Science, Office of
Advanced Scientific Computing Research, Accelerated Research in Quantum
Computing program.

\bibliographystyle{alpha}
\bibliography{paper}

\newpage
\begin{appendix}
\input{app-technical-claims.tex}

\input{app-timestamps.tex}
\input{pagerank.tex}
\input{app-index.tex}
\input{app-query-lb.tex}

\end{appendix}

\end{document}

%% file: intro.tex
%!TEX root = ./main.tex
\section{Introduction}

The random order streaming model for computation on graphs has been the focus
of much attention recently, resulting in truly sublinear algorithms for several
fundamental graph problems, i.e. algorithms whose space complexity is sublinear
in the number of \emph{vertices} (as opposed to edges) in the input
graph~\cite{KapralovKS14,CormodeJMM17,MonemizadehMPS17,PS18,KapralovMNT20}.
This is in sharp contrast to adversarially ordered streams, where $\Omega(n)$
space is often needed to solve even the most basic computational problems on
graphs~\cite{FeigenbaumKMSZ04}. This brings several fundamental problems on
graph streams (matching size, number of connected components, constant query
testable properties in bounded degree graphs) into the same regime as basic
statistical queries such as heavy hitters, frequency moment estimation and
distinct elements~\cite{AlonMS96}---problems that can be solved using space
polylogarithmic in the length of the stream.

Sampling random walks has numerous applications in large graph
analysis
(e.g.,~\cite{SpielmanT13,AndersenCL06,AndersenP09,CharikarOP03}), so
it has received quite a bit of attention in the adversarial streaming
model~\cite{SarmaGP11,Jin19,KPS0Y21}.  However, while these results
are useful for dense graphs, they all require $\Omega(n)$ space.

We show that the random order model allows us to break this barrier.  For
random order streams we give an algorithm that generates $\nwalks$ walks that
are $\e$-approximate to $k$-step random walks from uniformly random starting
vertices\footnote{Previous work has considered this problem when the start
vertex is adversarially chosen and given to the algorithm before processing the
stream. Unfortunately $o(n)$ space is impossible in this setting, as if the
start vertex is in, say, a two-edge path, finding the second edge of the path
will be difficult in the 50\% of cases it arrives after the first. For a formal lower bounds see
Appendix~\ref{app:chosenvertexlb}.}, using $(\frac1{\e})^{O(k)}\cdot
2^{O(k^2)}\cdot \nwalks$ words of space, independent of the graph size $n$. This
solves an open problem of Peng and Sohler on estimating return probabilities of
random walks~(\cite{PS18}, page 23). 

The exponential dependence on $\poly(k)$ here seems likely to be necessary, at
least up to the power of $k$, as recent work~\cite{CKKP21} has shown that
finding a length-$\ell$ component in a graph where every component is of length
at most $\ell$ requires $\ell^{\bOm{\ell}}$ space in a model
close\footnote{This lower bound applies when edges are grouped into pairs, and
the pairs arrive in a uniformly random order.  It does not
necessarily imply lower bounds on fully random-order streaming, but
it seems unlikely that this particular structure would make the problem
dramatically harder.} to random-order streaming. Performing a $k = \bTh{\ell^2}$
random walk from a randomly chosen vertex would suffice for this, and so we
expect any such algorithm needs at least $k^{\bOm{\sqrt{k}}}$ space.

Our algorithm immediately implies sublinear algorithms for graph analytics
based on short random walks, such as return probability estimation or
PageRank.  Consider PageRank with a constant reset probability
$\alpha$.  For a ``topic'' $T \subset V$---think, ``news websites'' or
``websites about gardening''---we can view the total PageRank of $T$
as a measure of the importance of that topic to the graph.  Our
walk sampling algorithm lets us estimate the total PageRank of $T$ to
within $\eps$ using $O_{\alpha, \eps}(1)$ space.

Our algorithmic results are for undirected graphs, because directed
graphs are hard: we show that sampling walks from directed graphs (or
just estimating PageRank) requires $\Omega(n)$ space in random order
streams.

\paragraph{Our results.} We now state our results formally.  We will
need the definition of $\e$-closeness of distributions below:
\begin{definition}[Pointwise $\e$-closeness of
distributions]\label{def:eps-close} We say that a distribution $p\in
\R_+^{\mathcal U}$ is $\e$-close pointwise to a distribution $q\in
\R_+^{\mathcal U}$ if for every $u\in \mathcal U$ one has $$
p(u)\in [1-\e, 1+\e]\cdot q(u).
$$
\end{definition}

We now define the notion of an $\e$-approximate sample of a $k$-step random walk:
\begin{definition}[$\e$-approximate sample]\label{def:eps-sample}
Given $G=(V, E)$ and a vertex $u\in V$ we say that $(X_0, X_1,\ldots, X_k)$ is an {\em $\e$-approximate sample} of the $k$-step random walk started at $u$ if 
the distribution of $(X_0, X_1,\ldots, X_k)$ is $\e$-close pointwise to the distribution of the $k$-step walk started at $u$ (see Definition~\ref{def:eps-close}).
\end{definition}

\paragraph{Main result.} Our main result is an algorithm for generating nearly independent $\e$-approximate samples of the $k$-step walk in the input graph $G$ presented as a randomly ordered stream:

\begin{theorem}\label{thm:main}
  There exists a constant $c'>0$ such that for every $n$-vertex graph $G$, for
  every $\e\in (n^{-1/1000}, 1/2)$, $\nwalks \leq n^{1/100}$, and $1\leq k\leq
  \min\left\{\frac{c'\log n}{\log (1/\e)}, \sqrt{c'\log n}\right\} $, the
  following holds:

  The output of $\walks(k, \e, \nwalks)$ (Algorithm~\ref{alg:simulate-walks}
  below) is $(n^{-1/100}+2^{-\nwalks})$-close in TV distance to the
  distribution of $\nwalks$ independent $\e$-approximate samples of the
  $k$-step random walk in $G$ started at vertices chosen uniformly at random.
  The space complexity of $\walks(k, \e, \nwalks)$ is upper bounded by
  $(1/\e)^{O(k)}\cdot 2^{O(k^2)}\cdot \nwalks$.
\end{theorem}

Using random walk sampling as a primitive, we give algorithms for two
important graph problems: computing the average return probability of
$k$-step random walks and estimating the PageRank of a subset of
nodes.

\paragraph{Return probability estimation.} For every integer $k\geq 1$ and
$u\in V$ let $p^k_u\in \R^V$ denote the distribution of the simple $k$-step
random walk started at $u$. The average return probability of $k$-step random
walks in $G$ is 
\begin{equation}
\text{rp}(G)=\frac1{n}\sum_{u\in V} p^k_u(u).
\end{equation}
We say that $\wh{\text{rp}}(G)$ estimates $\text{rp}(G)$ with precision $\e\in (0, 1)$ if 
\begin{equation}\label{eq:rp-approx}
|\text{rp}(G)-\wh{\text{rp}}(G)|\leq \e.
\end{equation}

\begin{remark}
Note that if the input graph $G$ consists of disjoint connected components with
mixing time bounded by $o(k)$, then $\text{rp}(G)$ is very close to the number
of connected components in $G$.  In particular, if every component in $G$ has
size at most $q$, the mixing times are
bounded by $q^{O(1)}$, so this gives another algorithm for approximately counting
the number of connected components in $G$ using space $2^{\text{poly}(q)}$, which is comparable to~\cite{PS18}.

In general, the average return probability can be viewed as a more robust
measure of connectivity than the number of components.
\end{remark}

Our algorithm for approximating the average return probability is \textsc{ApproxRP} (Algorithm~\ref{alg:approx-rp}  below).

\begin{algorithm}[H]
	\caption{\approxrp: approximate average $k$-step return probability over $u\in V$}  
	\label{alg:approx-rp} 

\begin{algorithmic}[1]
\Procedure{\approxrp($k, \e$)}{} \Comment{$k$ is the desired walk length,
$\e\in (0, 1)$ is a precision parameter}
\State $\nwalks \gets \frac{D}{\e^2}$ for a sufficiently large constant
$D>0$\label{line:s-def}
\State $(v^i)_{i\in \brac{\nwalks}}\gets \walks(k, \e, \nwalks)$ 
\State \Return $\frac{1}{\nwalks}\cdot |\{i:\in \brac{\nwalks}: v^i_k=v^i_0\}|$
\Comment{Empirical return probability}
\EndProcedure
\end{algorithmic}
\end{algorithm}

\begin{corollary}\label{cor:rp}
There exists a constant $c'>0$ such that for every graph $G=(V, E), |V|=n,
|E|=m$, for every $\e\in (n^{-1/1000}, 1/2)$  and $1\leq k\leq \min\left\{\frac{c'\log n}{\log (1/\e)}, \sqrt{c'\log n}\right\}
$ the following conditions hold.

Algorithm $\approxrp(\e, k)$ (Algorithm~\ref{alg:approx-rp}
below) computes an $\e$-approximation to the average return probability of a
$k$-step random walk in $G$ given as a random order
stream using space $(1/\e)^{O(k)}\cdot 2^{O(k^2)}$ with probability at least $9/10$
over the randomness of the stream and its internal randomness.
\end{corollary}

The proof of Corollary~\ref{cor:rp}  follows from Theorem~\ref{thm:main} by standard concentration inequalities and is presented in Appendix~\ref{app:pagerank}.

\paragraph{Estimating PageRank.} 
For a reset probability $\alpha\in (0, 1)$ the PageRank vector with reset probability $\alpha$, denoted by $p_\alpha\in \R^V$, satisfies
$$
p_\alpha=\alpha\cdot \frac{\mathbbm{1}}{n}+(1-\alpha) Mp_\alpha,
$$
where $M$ is the random walk transition matrix of $G$. We give an algorithm (\approxpagerank, Algorithm~\ref{alg:approx-pagerank} below) that, given a membership oracle for a subset $T\subseteq V$, computes an approximation $\wh{p}_\alpha(T)$ to $p_\alpha(T)=\sum_{u\in T} p_\alpha(u)$ such that 
\begin{equation}\label{eq:phat}
|\wh{p}_\alpha(T)-p_\alpha(T)|\leq \e.
\end{equation}

Our algorithm exploits the fact that PageRank with reset probability $\alpha$ is a mixture of distributions of random walks started with the uniform distribution over vertices of $G$  whose length follows the geometric distribution with parameter $\alpha$. Specifically,
\begin{equation*}
\begin{split}
p_{\alpha}&=(I-(1-\alpha) M)^{-1}\cdot \alpha\frac{\mathbbm{1}}{n}=\sum_{k\geq 0} \alpha (1-\alpha)^k M^k\cdot \frac{\mathbbm{1}}{n}\\
\end{split}
\end{equation*}
Therefore, an additive $\e$-approximation  to the PageRank $p_\alpha(T)$ of a set
$T$ as per~\eqref{eq:phat} can be obtained by truncating the sum above to its first
$O(\frac1{\alpha}\log(1/\e))$ terms and estimating them using \textsc{SimulateWalks}, our
algorithm for generating independent samples of random walks. This is exactly what
\approxpagerank{} (Algorithm~\ref{alg:approx-pagerank} below) does.

\begin{algorithm}[H]
	\caption{\approxpagerank: approximate PageRank with reset probability $\alpha\in (0, 1)$ on target set $T$ up to $\e$ additive error.}  
	\label{alg:approx-pagerank} 

\begin{algorithmic}[1]
\Procedure{\approxpagerank($\alpha, T, \e$)}{} \Comment{Approximate PageRank of
$T\subseteq V$}
\State\Comment{$\alpha$ is the reset probability, $\e$ is the additive precision}
\State \Comment{$T$ is given by a membership oracle}
\State $\nwalks \gets \frac{D}{\e^2}$ for a sufficiently large constant $D>0$
\State $(v^i)_{i\in \brac{\nwalks}}\gets \walks(\lceil\frac{2}{\alpha}\log(1/\e)\rceil,
\e, \nwalks\cdot \lceil \frac{2}{\alpha}\cdot \log(1/\e)\rceil)$ 
\State \Comment{Increase the number of sampled walks to account for different lengths}
\For{$j=0$ to $\lceil\frac{2}{\alpha}\log(1/\e)\rceil$}
\State $W(j)\gets$ walks $v^i$ with $i$ between $\nwalks \cdot j+1$ and
$\nwalks \cdot j$ truncated to $j$ steps.
\State \Comment{$W(j)$ are nearly independent collections of $s$ walks of
length $j$}
\EndFor
\State $\wh{p}\gets 0$
\For{$i=1$ to $\nwalks$}
\State $J\gets \text{Geom}(\alpha)$ \Comment{$J=j$ with probability $\alpha
(1-\alpha)^j$ for every integer $j\geq 0$}
\State {\bf If~}$J> \lceil\frac{2}{\alpha}\log(1/\e)\rceil$ {\bf ~then~continue}
\If {$i^{\text{th}}$ walk in $W(J)$ ends in $T$}\Comment{Use membership oracle for $T$}
\State $\wh{p}\gets \wh{p}+\frac{1}{\nwalks}$  
\EndIf
\EndFor
\State \Return $\wh{p}$
\EndProcedure
\end{algorithmic}
\end{algorithm}

\begin{corollary}\label{cor:pagerank}
There exists a constant $c'>0$ such that for every graph $G=(V, E), |V|=n, |E|=m$, for every $\alpha\in (0, 1)$, every $\e\in \left(2^{-o(\sqrt{\log n})}, 1/2\right)$  such that $\frac{1}{\alpha}\leq \frac{\sqrt{\log n}}{4\log (1/\e)}$ the following conditions hold.

For every $T\subseteq V$ \approxpagerank{} (Algorithm~\ref{alg:approx-pagerank})  approximates $p_\alpha(T)$ for constant $\alpha\in (0, 1)$ up to additive error $\e$ with probability at least $9/10$ using $(1/\e)^{O(\frac1{\alpha} \log(1/\e))}\cdot 2^{O(\frac1{\alpha^2}\log^2(1/\e))}$ space given a randomly ordered stream of edges of $G$, assuming a membership oracle for the target set $T$.
\end{corollary}

The proof of Corollary~\ref{cor:pagerank} follows from Theorem~\ref{thm:main} by standard concentration inequalities and is presented in Appendix~\ref{app:pagerank}.

\paragraph{Lower bounds for directed graphs.} PageRank was first studied for
\emph{directed} graphs, and so it is natural to ask if it is possible to extend
these algorithms to that setting. We show that it is not, and in fact both
sampling from the random walk distribution and approximating the PageRank of a
vertex set in a directed graph require $\bOm{n}$ bit of storage. This holds
even if we restrict to approximating the distribution of very short random
walks.
\begin{restatable}{theorem}{dgrwlb}
\label{thm:dgrwlb}
For any constant $\varepsilon < 1/4$, the following holds for all $k \ge 3$ and
all $n$: there is a family of directed graphs with no more than $n$ vertices
and edges such that any random order streaming algorithm that
$\varepsilon$-approximates the distribution of length-$k$ random walks on
graphs drawn from the family uses $\bOm{n}$ bits of space, with a constant
factor depending only on $\varepsilon$.
\end{restatable}
\begin{restatable}{theorem}{dgprlb}
\label{thm:dgprlb}
Let $\alpha$ be a given (constant) reset probability for PageRank. For any
constants $\varepsilon < (1 - \alpha)^3 - \frac{1}{2}$, $\delta < 1/4$, the
following holds for all $n$: there is a family of directed graphs with no more
than $n$ vertices and edges such that any random order streaming algorithm that
returns a $\varepsilon$ additive approximation to the PageRank of vertex sets
in these graphs with probability $1 - \delta$ uses  $\bOm{n}$ bits of space,
with a constant factor depending only on $\alpha$, $\varepsilon$, and $\delta$.
\end{restatable}
Note that this second lower bound applies only when $(1 - \alpha)^3 > 1/2$,
i.e.\ $\alpha$ must not be much greater than $0.2$. In applications $\alpha$ is
typically $0.15$, so this is a reasonable assumption.

\subsection{Algorithmic Techniques}
In what follows we discuss the main challenges involved in obtaining
Theorem~\ref{thm:main} and the key ideas behind our approach.  A natural
approach would be to sample a large collection of vertices in the graph
uniformly at random, and to then try simulating $k$-step walks from them using
the random order of the stream.  Most of these simulations will fail, but our
hope is that with a reasonably large (specifically,
$\varepsilon^{O(k)}2^{-O(k^2)}$) probability $\tau$, we will successfully
output a nearly-uniform random walk from each start vertex, and that, when we
don't succeed, we will know that we failed. 

How could we simulate a walk from a given vertex $v$ using the random
stream?  One idea might be: starting with $v$, repeatedly take the
next edge incident on the current vertex, and stop after $k$ steps (or
output $\bot$ if the stream terminates before the $k^{\text{th}}$ step).  The
intuition is that a random walk has a $1/k!$ chance of appearing in
the stream in order, but this approach has two problems.  First, it
never traverses the same edge twice, while a random walk has a good
chance of doing so whenever it encounters a low-degree vertex.
Second, even for walks where every edge is distinct, the probability
that it outputs a given path is not proportional to the probability of
that path arising as a random walk. For example, when starting at the
endpoint of a path, a length-2 walk occurs with probability $1/2$ and
is sampled with probability $1/2$ (if the adjacent edge precedes the
next edge out); but when starting in the middle of a path, the
\emph{first} adjacent edge is more likely to precede the next edge
out, increasing the probability of sampling a length-2 walk to
$2/3$.

\paragraph{Repeated edges.} We fix the first issue by associating each walk
$(e_1, e_2,\ldots, e_k)$ with a template $\pi$, a sequence of numbers
that encodes the amount of `backtracking' that the walk
$(e_1,e_2,\ldots, e_k)$ does. Formally, we define:
\begin{definition}[Walk template]\label{def:template}
A {\em $k$-walk template} is a tuple $\pi\in \Pi_k$, where $\Pi_k:=[1]\times [2]\times [3]\times \ldots \times [k]$. 
\end{definition}
\begin{definition}[A walk conforming with a template]\label{def:conforms}
We say that a walk $e=(e_1,e_2,\ldots, e_k)\in E^k$ {\em conforms with a
template $\pi\in \Pi_k$}  if for every $j\in [k]$ one has 
$$
\pi_j=\min\{i\in [k]: e_i=e_j\}.
$$
Similarly, for $\ell\in [k]$ we say that a walk $e=(e_1,e_2,\ldots, e_\ell)\in E^\ell$ {\em conforms with a
template $\pi\in \Pi_k$}  if for every $j\in [\ell]$ one has 
$$
\pi_j=\min\{i\in [\ell]: e_i=e_j\}.
$$
\end{definition}
Note that a walk $(e_1, e_2,\ldots, e_k)$
conforms with a template $\pi\in \Pi_k$ if and only if for every $j\in [k]$,
$\pi_j$ is the first time the edge $e_j$ appears in the walk.

Because every walk $(e_1,e_2,\ldots, e_k)$ conforms with exactly one
template, our random walk generation procedure can proceed by first
sampling a template uniformly at random, and then generating a walk
that conforms with that template.

\paragraph{Debiasing the estimate.} To address the second issue with
the naive approach---that the probability we find a given walk is not
proportional to the random walk probability---we modify the walk
procedure slightly to not always follow the next edge out of each
vertex.  Instead, we choose our first edge out of the start vertex $v$
uniformly at random from the edges incident to $v$ in the first $\eta$
fraction of the stream, for a small parameter $\eta$; our second edge
is chosen uniformly from the edges incident to this vertex in the second $\eta$
fraction of the stream, and so on.  But this is subject to conforming to the
template---in step $j$, if $\pi_j \neq j$, we ignore the $j^{\text{th}}$ $\eta$ fraction
of the stream and instead reuse the $\pi_j^{\text{th}}$ edge we've already taken.

With this approach, we can correct differences in the probability of finding
each walk. The probability that we find a given walk is (i) $1/k!$, the
probability that we sample the right template, times (ii) an $\eta$ factor for
each distinct edge in the walk, the probability that the stream is such that it
is \emph{possible} for our algorithm's random choices to find this walk, times
(iii) the probability that our algorithm makes the correct choices to find the
walk.  This last probability depends on the stream, being the product over
steps of the inverse number of edges incident to the previous vertex in the
appropriate $\eta$ fraction of the stream.  The key is that this probability
$p$ is known after we see the stream; so if we knew the \emph{true} random walk
probability $q$ for this walk, we could rejection sample with probability
proportional to $q/p$ to output walks under the correct distribution.

So how can we estimate the correct probability $q$ of a given random walk with
small expected error?  The random walk probability $q$ is
$\prod_{j=1}^k \frac{1}{d_{j-1}}$, where $d_{j-1}$ is the degree of
the $(j-1)^{\text{th}}$ vertex in the walk (with $d_0$ being the starting degree).
We can estimate each $d_j$ by watching the stream after we finish the
sampling procedure; this will contain a $(1 - k \eta)$ fraction of the stream.
One might expect this to introduce an error of about $(1-k \eta)^k \approx
e^{-k^2 \eta}$ in our estimate of $q$, but in fact the error can be much larger
because a vertex may be visited as many as $\bTh{k}$ times in a walk. 

For a constant-degree vertex, there is an $\eta k$ chance that we will miss at
least one of its edges, in which case our estimate will be off by a constant
factor, which could in turn lead to a $2^{\Theta(k)}$ relative error in $q$ if
the vertex appears $\bTh{k}$ times in the walk, for $\eta 2^{\Theta(k)}$
expected relative error.  For this reason we need to set $\eta <
2^{-\Theta(k)}$, which leads to the final $2^{-O(k^2)}$ term in $\tau$, and
thus the $2^{\bO{k^2}}$ term in our space complexity.

\paragraph{Repetition and near-independence.}
The above argument leads to an $O(k)$-space algorithm that, with probability
$2^{-O(k^2)}$, outputs a nearly uniform random walk.  To make this useful, we
need to repeat it at least $s = 2^{O(k^2)}$ times so that we may actually find
walks.  The challenge here is that these are not independent repetitions: the
output of the algorithm depends on the random order of the stream, which is
shared by each copy of the algorithm.

Fortunately, the repetitions are \emph{nearly} independent.  Knowing the path
taken in a given attempt to sample a walk tells us something about the arrival
times of the other edges incident to the vertices visited in that walk, but it
is independent of edges not incident to vertices on the path.  If the graph had
no high degree vertices---say, the maximum degree were $n^{1/4}$---this would
be sufficient: the probability that any given walk visits a degree-$d$ vertex
is at most $kd/n < n^{-.7}$, so if $s < n^{.01}$ we will probably never visit
two adjacent vertices.

For high-degree vertices we need a different analysis: a high-degree vertex $v$
will with high probability have many edges as possibilities in each stage, so
knowing the behavior of $s$ other walks only has a small effect on which edges
are likely to be followed after visiting $v$.  Formally, we introduce a hybrid
algorithm where the behavior on high degree vertices is independent of the
stream, show that the distribution of the output of the original algorithm is
close to that of the hybrid algorithm, then use the above argument for low
degree vertices.

\subsection{Lower Bound Techniques}
Our lower bound is based on the following property of directed graphs: if a
vertex has high \emph{in}-degree but low \emph{out}-degree, that vertex can
cause the vast majority of random walks in a graph to be ``channeled'' into one
path. If this path has multiple edges any algorithm that estimates the random
walk distribution or the PageRank vector will need to observe all of them,
which is inherently difficult as later edges in the path will not be recognized
as significant unless they arrive after all the earlier edges. This is depicted
in Figure~\ref{fig:hardgraph}.

\input{hard-graph.tex}

Formally, we give a method for encoding an instance of the Indexing problem in
a graph stream, similar to techniques used in~\cite{CCM16}. In the Indexing
problem, Alice has an $n$ bit string $x$ and Bob an index $I$.  Alice must send
Bob a message that allows him to guess the value of $x_I$. It is known that
Alice must send $\Omega(n)$ bits if she wants to succeed at this task.

Ordinarily it is difficult to encode communication problems as random order
graph streams, as the fact that the edges may arrive in any order makes it
difficult to assign parts of the graph to different players. We evade this
difficulty by making use of the fact that the indexing problem is hard even if
the players are guaranteed a uniform distribution on their inputs \emph{and}
only have to succeed with probability $1/2 + \varepsilon$ for some constant
$\varepsilon$. 

In our method, Bob encodes his index as a single edge from a high in-degree
vertex, and Alice encodes her bits as $n$ edges that it might point to, with
each pointing to a vertex representing $0$ or $1$. Therefore, almost all random
walks in the graph end at a vertex representing $x_I$, and a large fraction of
the PageRank vector's weight will be on this vertex.

As the edges are required to arrive in a uniformly random order, sometimes
Alice's edges may arrive after Bob's edge, in which case Bob is responsible for
inserting them in the stream.  them. In that case he simply guesses what they
should be. This means that half the time the graph will encode a random answer
rather than a solution to Indexing, but this is still sufficient for him to
succeed more than half the time. This encoding is illustrated in
Figure~\ref{fig:roindexencoding}.

\input{index-encoding.tex}

\subsection{Related Work}

For adversarial streams the problem of generating a $k$-step walk out of a given starting vertex was first considered in a paper of~\cite{SarmaGP11}, where it is shown how to generate $n$ walks of length $k$ using $\wt{O}(\sqrt{k})$ passes over the stream and $\wt{O}(n)$ space. The work of~\cite{Jin19} gives a single pass algorithm with space complexity $\wt{O}(n\cdot \sqrt{k})$ undirected graphs, and shows that this is best possible. Another recent work~\cite{KPS0Y21} gives two-pass algorithms for generating walks of length $k$ in general (even directed) graphs using $\wt{O}(n\cdot \sqrt{k})$ space, which they also show it essentially best possible. 

%KonradMM12,Bernstein20,AssadiB21,

%% file: hard-graph.tex
\begin{figure}
    \centering
    \tikzset{basevertex/.style={shape=circle, line width=0.5,
    minimum size=4pt, inner sep=0pt, draw}}
    \tikzset{defaultvertex/.style={basevertex, fill=blue!70}}
    \begin{tikzpicture}[%
            VertexStyle/.style={defaultvertex},
            fat arrow/.style={single arrow,
                thick,draw=blue!70,fill=blue!30,
            minimum height=10mm},
        scale = 1]

\begin{scope}[thick,decoration={
    markings,
        mark=at position 0.5 with {\arrow{>}}}
            ] 
        \node[style={defaultvertex}](u) at (4,{0.4 * 3.5}) {};
        \node at (4,{0.4 * 3.5 - 0.3}) {$u$};
        \foreach \i in {0,...,9} {
            \node[style={defaultvertex}](in) at (0,{0.4 * \i}) {};
            \draw[postaction=decorate] (in) -- (u);
        }
        \node[style={defaultvertex}](v) at (8, {0.4* 3.5}) {};
        \node at (8,{0.4 * 3.5 - 0.3}) {$v$};
        \draw[postaction=decorate] (u) -- (v);
        \node[style={defaultvertex}](w) at (12, {0.4* 3.5}) {};
        \node at (12,{0.4 * 3.5 - 0.3}) {$w$};
        \draw[postaction=decorate] (v) -- (w);
        \foreach \i in {-0.5, 1.5, 5.5, 7.5} {
            \node[style={defaultvertex}] (v) at (8, {0.4 * \i}) {};
            \node[style={defaultvertex}] (w) at (12, {0.4 * \i}) {};
            \draw[postaction=decorate] (v) -- (w);
        }
\end{scope}
    \end{tikzpicture}
    \caption{A directed graph that makes it hard to find walks. Most walks in
    the graph start at one of the vertices on the left and then go into $u$,
    then $v$, then $w$. But if $vw$ arrives before $uv$ it is impossible to
    know that it should be part of all of these walks.}
    \label{fig:hardgraph}
\end{figure}
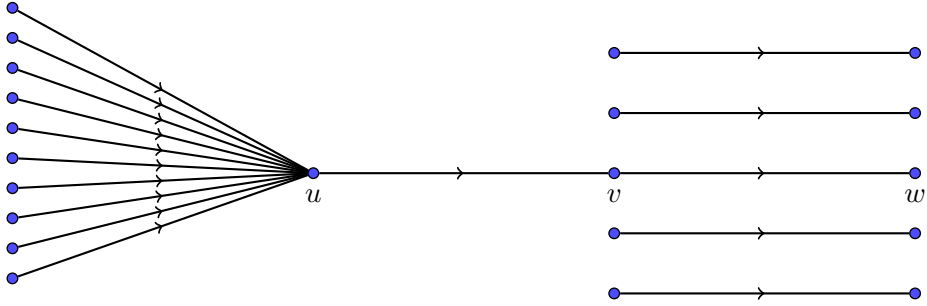

%% file: index-encoding.tex
\begin{figure}
    \centering
    \begin{subfigure}[t]{\textwidth}
        \centering
        \tikzset{basevertex/.style={shape=circle, line width=0.5,
        minimum size=4pt, inner sep=0pt, draw}}
        \tikzset{defaultvertex/.style={basevertex, fill=blue!70}}
        \begin{tikzpicture}[%
                VertexStyle/.style={defaultvertex},
                fat arrow/.style={single arrow,
                    thick,draw=blue!70,fill=blue!30,
                minimum height=10mm},
            scale = 1]

            \begin{scope}[thick,decoration={
                    markings,
                    mark=at position 0.57 with {\arrow{>}}}
                ] 
                \node[style={defaultvertex}](u) at (4,{0.4 * 3.5}) {};
                \foreach \i in {0,...,9} {
                    \node[style={defaultvertex}](in) at (0,{0.4 * \i}) {};
                    \draw[postaction=decorate] (in) -- (u);
                }
                \foreach \i in {0,...,9} {
                    \node[style={defaultvertex}](v\i) at (8, {0.4* \i}) {};
                }
                \node at (10.3, 0.5) {$x$};
                \draw[postaction=decorate] (u) -- (v3);
                \node at (6.17, 0.8) {$I$};
                \node[style={defaultvertex}](d0) at (12, {0.4* 3}) {};
                \node[style={defaultvertex}](d1) at (12, {0.4* 5}) {};
                \foreach \i in {1, 3, 4, 6, 7} {
                    \draw[postaction=decorate](v\i) -- (d0);
                }
                \foreach \i in {0, 2, 5, 8, 9} {
                    \draw[postaction=decorate](v\i) -- (d1);
                }
                \node[style={defaultvertex}](e0) at (13, {0.4* 3}) {};
                \node[style={defaultvertex}](e1) at (13, {0.4* 5}) {};
                \draw[postaction=decorate] (d0) to [out=-20, in=-160] (e0);
                \draw[postaction=decorate] (e0) to [out=160, in=20] (d0);
                \draw[postaction=decorate] (d1) to [out=-20, in=-160] (e1);
                \draw[postaction=decorate] (e1) to [out=160, in=20] (d1);
                \node at (13.5, {0.4* 3}) {1};
                \node at (13.5, {0.4* 5}) {0};
            \end{scope}
        \end{tikzpicture}
        \caption{Encoding an instance of Indexing in a directed graph. The high
            in-degree vertex will have one outgoing edge, encoding Bob's index $I$, which
            will point to one of $n$ vertices, each of which encodes one of Alice's bits $x$
        by pointing to one of two corresponding loops.}
    \end{subfigure}
    \begin{subfigure}[t]{\textwidth}
        \centering
        \tikzset{basevertex/.style={shape=circle, line width=0.5,
        minimum size=4pt, inner sep=0pt, draw}}
        \tikzset{defaultvertex/.style={basevertex, fill=blue!70}}
        \begin{tikzpicture}[%
                VertexStyle/.style={defaultvertex},
                fat arrow/.style={single arrow,
                    thick,draw=blue!70,fill=blue!30,
                minimum height=10mm},
            scale = 1]

            \begin{scope}[thick,decoration={
                    markings,
                    mark=at position 0.57 with {\arrow{>}}}
                ] 
                \foreach \i in {0, 1, 5, 6, 7, 8} {
                    \node[style={defaultvertex}](v\i) at (8, {0.4* \i}) {};
                }
                \node[style={defaultvertex}](d0) at (12, {0.4* 3}) {};
                \node[style={defaultvertex}](d1) at (12, {0.4* 5}) {};
                \foreach \i in {1, 6, 7} {
                    \draw[postaction=decorate](v\i) -- (d0);
                }
                \foreach \i in {0,  5, 8} {
                    \draw[postaction=decorate](v\i) -- (d1);
                }
                \node at (10.3, -0.5) {$x_{\set{0, 1, 5, 6, 7, 8}}$};
                \node[style={defaultvertex}](u) at (13,{0.4 * 3.5}) {};
                \node[style={defaultvertex}](v3) at (17, {0.4* 3}) {};
                \draw[postaction=decorate] (u) -- (v3);
                \node at (15.17, -0.5) {$I$};
                \foreach \i in {2,3,4,9} {
                    \node[style={defaultvertex}](v\i) at (18, {0.4* \i}) {};
                }
                \node[style={defaultvertex}](d0) at (22, {0.4* 3}) {};
                \node[style={defaultvertex}](d1) at (22, {0.4* 5}) {};
                \foreach \i in {2} {
                    \draw[postaction=decorate](v\i) -- (d0);
                }
                \foreach \i in {3,4,9} {
                    \draw[postaction=decorate](v\i) -- (d1);
                }
                \node at (20.3, -0.5) {\color{red} $y_{\set{2,3,4,9}}$};
            \end{scope}
        \end{tikzpicture}
        \caption{Converting the encoding into a random-order stream. The edges
        that do not depend on Alice or Bob's input can be inserted using shared
        randomness, and so are ignored here.
        Those of Alice's edges that should arrive \emph{before} Bob's edge are
        encoded as normal, but those that arrive after are instead inserted by
        Bob by chosing a random guess $y$ for Alice's input. Half of the time
        the edge encoding $x_I$ arrives before the edge encoding $I$, and the
        other half of the time there is still a $0.5$ chance of Bob guessing
        correctly, so the encoding is ``correct'' with probability $0.75$.}
    \end{subfigure} 
    \caption{Encoding an instance of Indexing as a random-order graph stream.}
    \label{fig:roindexencoding}
\end{figure}
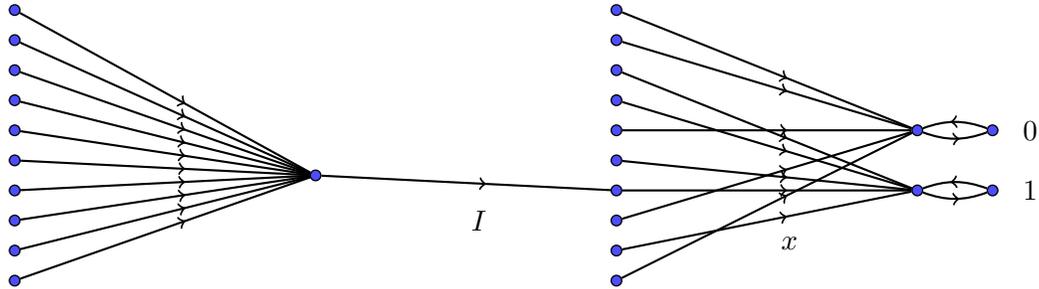
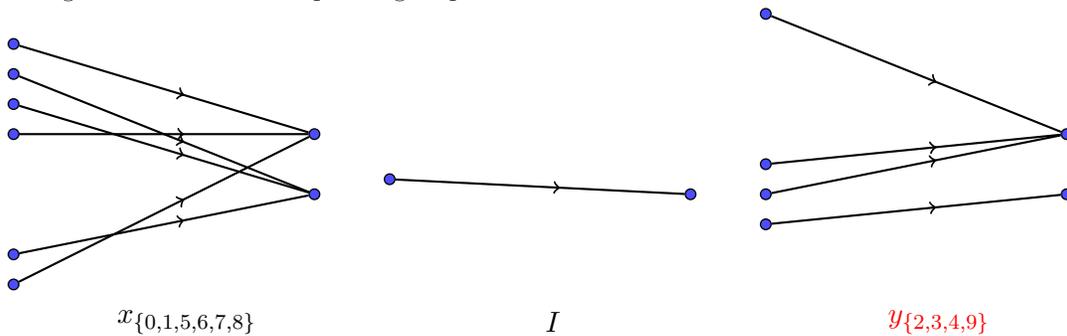

%% file: rp.tex
%!TEX root = ./main.tex

\newcommand{\Ev}{\mathcal{E}}

\section{Basic Definitions and Claims}
\paragraph{Basic notation.} For any integer $a\geq 1$ we write $[a]=\{1,
2,\ldots, a\}$. For any set $S$ we use $\Uc\paren{S}$ to denote the uniform
distribution on $S$. For a pair of (discrete) random variables $X$ and $Y$ we
let $d_{TV}(X, Y)$ denote the total variation distance between $X$ and $Y$,
which equals one half of the $\ell_1$ distance between their distributions. For
a vertex $v\in V$ we write $d(v)$ to denote the degree of $v$ in $G$, and
$\delta(v)$ to denote the set of its incident edges. For an integer $k\geq 0$
and two vertices $u, v\in V$ we write $p^k_v(u)$ to denote the probability that
the simple random walk started at $v$ reaches $u$ after $k$ steps. We assume
for simplicity that the graph does not have isolated vertices (all our
algorithms can be easily adapted to handle isolated vertices, so this is
without loss of generality).

\subsection{Stream Model}\label{sec:stream-model}
We assume that we receive the edges of the graph in a uniformly random order.
We think of this order as being generated by assigning each edge $f \in E$ a
``timestamp'' $t_f \in \brac{0,1}$ uniformly at random. The edges are then
presented to the algorithm in ascending order of timestamp.

Typically a random-order stream does not come with such timestamps. However, as
we show in Appendix~\ref{app:streamstamps}, an algorithm can generate
timestamps distributed appropriately using just $\bO{\log n}$ extra space, for
any desired $1/\poly(n)$ accuracy (this will suffice, as we can set the accuracy
to e.g.~$1/n^{100}$ and with very high probability the output of the algorithm
will not be influenced by changing any number of timestamps by that much). 

We assume knowledge of $m$, the number of edges in the graph. This assumption, however, can be removed by `guessing' the right value of $m$ (by running several copies of the algorithm in parallel), at the expense of a mild loss in the space complexity.

\subsection{Algorithm}
Our walk sampling algorithm (\walksreset,
Algorithm~\ref{alg:walks-reset} below) is quite simple: it samples
a large enough set of nodes uniformly at random, together with
independent uniformly random templates for every sampled node. It then
runs our random walk generation procedure (\wft)
from each such node, and outputs a sample of the runs that do not
return FAIL (i.e., those invocations of \wft{} that
terminate with a valid walk).

\begin{algorithm}[H]
	\caption{\walksreset: simulate $s$ samples of $k$-step walk started from uniformly random vertices, with reset (i.e., allowing walks to fail)}  
	\label{alg:walks-reset} 

\begin{algorithmic}[1]
\Procedure{\walksreset($k, \e, s$)}{} 
\State \Comment{$k$ is the desired walk length, $\e\in (0, 1)$ is a precision parameter}
\State $\eta \gets \e^{8}\cdot 2^{-Ck}$\Comment{For a large enough constant $C>0$} \label{line:eta-def}
\State $v^i_0\gets $ independent uniform sample from $V$ for $i\in [s]$
\For{$i\in [s]$}
\State Choose $\pi^i\sim \mathcal{U}(\Pi_k)$ \Comment{$\Pi_k$ is the set of walk templates of length $k$, see Definition~\ref{def:template}}
\State $v^i\gets \wft(v^i_0, \pi^i, k, \eta)$  \Comment{Run in parallel on the same stream}
\EndFor
\State \Return $(v^i)_{i\in [s]}$
\EndProcedure
\end{algorithmic}
\end{algorithm}

\begin{algorithm}[H]
	\caption{\walks: simulate $a$ samples of $k$-step walk started from uniformly random vertices}  
	\label{alg:simulate-walks} 

\begin{algorithmic}[1]
\Procedure{\walks($\e, k, \nwalks$)}{} 
\State \Comment{$k$ is the desired walk length, $\e\in (0, 1)$ is a precision parameter}
\State $s\gets \nwalks\cdot 100\eta^{-k}\cdot k!$ \Comment{Increase number of starting nodes to account for failed walks}
\State $(v^i)_{i\in [s]}\gets \walksreset(k, \e, s)$
\If{at least $\nwalks$ of the walks $(v^i)_{i\in [s]}$ succeeded}
\State \Return the first $\nwalks$ successful walks 
\Else
\State \Return FAIL
\EndIf
\EndProcedure
\end{algorithmic}
\end{algorithm}

\paragraph{Overview of random walk generation
({\normalfont\wft}).} Our main random walk generation
procedure is \wft{} (Algorithm~\ref{alg:walk-from-template}
below). The procedure gets as input a starting vertex, a template $\pi$, the
target length of the walk and a parameter $\eta$ that corresponds to the
fraction of the stream that is used to generate a single step of the walk.
Setting $\eta$ small lets us limit various correlations, but hurts the
performance, since the probability of \wft{} terminating
with a valid walk (as opposed to outputting FAIL) is about $\eta^k$. 

The procedure \wft{} itself is natural: it partitions the first $k\cdot \eta$ fraction of the stream
into intervals of length $\eta$, and for every $j\in [k]$ either it uses the
$j^{\text{th}}$ interval to sample a new edge (if the template $\pi$ prescribes this,
i.e. satisfies $\pi_j=j$; see line~\ref{line:hj} of
Algorithm~\ref{alg:walk-from-template}), or it takes the corresponding previously traversed edge if the template $\pi$ prescribes this, i.e.
satisfies $\pi_j<j$, and the corresponding edge is in the neighborhood of the current vertex (see line~\ref{line:back-edge} of
Algorithm~\ref{alg:walk-from-template}). The sampling is done using reservoir
sampling, and therefore is space efficient. 

After $k$ phases, having constructed a candidate walk $(f_1, f_2, \ldots,
f_k)$, the algorithm first uses the empirical batches to infer a partition of
the candidate walk into batches, and then uses this information to perform rejection
sampling. The goal of rejection sampling is to reduce the probability of
picking up a walk to be proportional to the product of the inverse degrees of the
first $k$ vertices in the walk, i.e. to the correct probability.  To achieve
this \wft{} maintains the probability $p_k$ of having
collected the walk $(f_1,\ldots, f_k)$. The algorithm then uses the remainder
of the stream to compute estimates $(\wh{d}_{j})_{j = 0}^{k-1}$, of the degrees
of the vertices on the candidate walk, and keeps the candidate walk with
probability proportional to $p_k^{-1}\prod_{j=1}^{k-1}(1/\wh{d}_{j})$ (see
line~\ref{line:rejection-sampling} of Algorithm~\ref{alg:walk-from-template}).

The algorithm is formally described as Algorithm~\ref{alg:walk-from-template}
below. For a stream $\sigma$ and parameters $\alpha, \beta\in [0, 1]$ we define
$$
\sigma[\alpha, \beta)=\left\{e\in E: t_e\in [\alpha, \beta)\right\}.
$$
Recall that we think of every edge $e\in E$ as being assigned an independent uniformly random timestamp $t_e\in [0, 1]$, and the edges being presented in increasing order of these timestamps (see Section~\ref{sec:stream-model} for more discussion of this assumption).

\begin{algorithm}
	\caption{\wft: generate a random walk from starting vertex $u_0$ that conforms with a template $\pi$}  
	\label{alg:walk-from-template} 

\begin{algorithmic}[1]
\Procedure{\wft($u_0, \pi, k, \eta$)}{} \Comment{$k$ is the desired walk length, $u_0$ is the starting vertex}
%\State $p_0\gets 1$ 
\For{$j=1$ to $k$}
\State $u\gets u_{j-1}$
\If{$\pi_j=j$}
\State $H_j\gets $ edges in $\delta(u)\cap \sigma[\eta \cdot (j-1), \eta \cdot j)$\label{line:hj}
\State {\bf If~}{$H_j=\emptyset$} {\bf~then~} \Return FAIL 
\State $f_j \gets \mathcal{U}(H_j)$ \Comment{Implemented using reservoir
sampling}
\State $\gamma_j \gets \frac1{|H_j|}\cdot \eta$
\Else
\State {\bf If~}{$f_{\pi_j}\not \in \delta(u)$} {\bf~then~} \Return FAIL
\State $f_j\gets f_{\pi_j}$\label{line:back-edge}
\State $\gamma_j \gets 1$
\EndIf
%\State $p_j\gets p_{j-1} \cdot \gamma_j$
\State $u_j\gets $ endpoint of $f_j$ other than $u$ 
\EndFor
\For{$j\in [k]$} \Comment{Compute degree estimates for vertices on the walk}
\State $\wh{d}_{j-1}\gets$ degree of $u_{j-1}$ in $\{f_1, \ldots, f_k\}\cup \sigma[\eta \cdot k, 1]$ \label{line:deg-est}
\EndFor
\State $\alpha \gets \prod_{j\in [k]} \min(\frac{\eta }{\gamma_j \wh{d}_{j-1}}, 1)$ \Comment{Done in postprocessing}\label{line:out}
\State \Return $(u_0, \ldots, u_k)$ with probability $\alpha$ and \Return FAIL otherwise \label{line:rejection-sampling}
%\State \Comment{$q/p_k<1$ with high probability}
\EndProcedure
\end{algorithmic}
\end{algorithm}

In what follows we prove that, if we set our parameters appropriately, each
walk is output with almost the correct probability.

We assume that parameters $n, k$ and $\eta$ satisfy the following:

\begin{enumerate}[label=\textbf{(P\arabic*)}]
\item $k\leq c\log n/\log\log n$ for a small constant $c>0$ \label{p1}
\item $\eta\in (n^{-1/100},  2^{-Ck})$ for a sufficiently large constant $C\geq 8$ \label{p2}
\end{enumerate}

\section{Analysis of \wft}

\begin{restatable}{lemma}{walkdistribution}\label{lm:walk-distribution}

For every integer $k\geq 1$ and real number $\eta$ satisfying~\ref{p1}
and~\ref{p2}, and every $\pi\in \Pi_k$, the following holds:

For every $v\in V$ and every length $k$ walk $\mathbf{v}=(v_0, v_1,\ldots, v_{k-1}, v_k)$ from $v_0 = v$, an invocation of \wft$(v, \pi, k, \eta)$ (Algorithm~\ref{alg:walk-from-template}) outputs $\mathbf{v}$ with probability 
$$
p\in \left[ 1, 1+\eta^{1/7} \right]\cdot \eta^k\cdot\prod_{j\in
[k]} \frac1{d(v_{j-1})},
$$
if $\mathbf{v}$ conforms with $\pi$ and with probability $0$ otherwise.
\end{restatable}

The following corollary is an immediate consequence of Lemma~\ref{lm:walk-distribution}:
\begin{corollary}\label{cor:walk-distribution}
For every integer $k\geq 1$ and real number $\eta$ satisfying~\ref{p1}
and~\ref{p2},  the following holds for $\pi$ sampled from $\Uc(\Pi_k)$:

For every $v\in V$ and every length $k$ walk $\mathbf{v}=(v_0, v_1,\ldots,
v_{k-1}, v_k)$ from $v_0 = v$,  an invocation of \wft$(v,
\pi, k, \eta)$ (Algorithm~\ref{alg:walk-from-template}) outputs $\mathbf{v}$
with probability
$$
p\in \left[ 1, 1+\bO{\eta^{1/7}} \right]\cdot \frac1{k!}\cdot
\eta^k\cdot\prod_{j\in [k]} \frac1{d(v_{j-1})}.
$$
\end{corollary}
\begin{proof}
Per Definition~\ref{def:conforms} there exists a unique template $\pi'\in \Pi_\ell$  that
 $\mathbf{v}$ conforms with. The walk $\mathbf{v}$ conforms with any extension $\pi$ of $\pi'$ to a template in $\Pi_k$, and no other template. For every such $\pi$ the corresponding invocation of \wft$(v_0,
\pi, k, \eta)$  constructs $\mathbf{v}$ at the end of the first $\ell$ iterations of its main loop with probability
\begin{equation}\label{eq:p}
p\in \left[ 1, 1+\bO{\eta^{1/7}} \right]\cdot \eta^\ell\cdot\prod_{j\in
[\ell]} \frac1{d(v_{j-1})}
\end{equation}
by Lemma~\ref{lm:walk-distribution}, and with probability zero for other $\pi$. 
Thus, the result follows since $\pi$ is selected uniformly at random from $\Pi_k$. 
\end{proof}

\subsection{Useful Technical Results}
The following are results that will be useful in our analysis. Proofs are
deferred to Appendix~\ref{app:technicalproofs}.

Recall that for integer $k\geq 0$ and two vertices $u, v\in V$ we write
$p^k_v(u)$ to denote the probability that the simple random walk started at $v$
reaches $u$ after $k$ steps.
\begin{restatable}{claim}{average}
\label{cl:average}
Let $v\in V$ be chosen uniformly at random in a graph with no isolated vertices. Then for every $u\in V$ and $k\geq 0$ one has $\expect_{v\sim \mathcal{U}(V)} [p^k_v(u)]\leq d(u)/n$.
\end{restatable}

The following is a consquence of Bennett's inequality.
\begin{restatable}{lemma}{concentration}
\label{lm:concentration}
Let $Y=\sum_i \alpha_i X_i$,  where $\alpha_i\in \{0, 1\}$ and $X_i\sim
\text{Ber}(\eta)$ are independent for some $\eta\in (0, 1/50)$. Then for every
$d\geq \sum_i \alpha_i$
$$
\Pr[Y\geq d/2]\leq (3\eta)^{d/5}.
$$
\end{restatable}

\begin{restatable}{lemma}{productlm}
\label{lm:prod}
 Let $E_1, \dotsc, E_k, Z_1, \dotsc, Z_k$ be arbitrarily correlated random
 variables. Let  $\wt{\eta}\in (0, e^{-5k})$ and the positive integers
 $(q_i)_{i=1}^k$ be such that, for all $i \in \brac{k}$, $E_i \in \{0, 1\}$,
 $\expect[E_i] \leq \wt{\eta}^{q_i/5}$, $Z_i \in [0, 1]$, and $\expect[Z_i]
 \leq \wt{\eta}$. Then
 \[
    \expect\left[\prod_{i=1}^k (1 + Z_i + q_i E_i )\right] \leq 1 + 3^k\wt{\eta}^{1/5}.
  \]
\end{restatable}

\subsection{Proof of Lemma~\ref{lm:walk-distribution}}

We now prove Lemma~\ref{lm:walk-distribution}, restated here for convenience of the reader.

\walkdistribution*

\begin{proof}
  There are three sources of randomness in \wft\
  (Algorithm~\ref{alg:walk-from-template}): the stream $\sigma$, the reservoir
  sampling $r$ to find the path, and the rejection sampling at the end.  We
  first analyze the event $\Fc$ that a given walk $\mathbf{v}$ is ``collected'',
  meaning that is found but might be rejected in the final rejection sampling
  step. We use $F$ to denote the indicator variable associated with $\Fc$.

  For $j\in [k]$ let $e_j=(v_{j-1}, v_j)$.  Let
  \begin{equation}\label{eq:gamma-def}
    \Gamma:=\left\{\sigma \text{ a stream }: \forall j\in [k]\text{~such that~}\pi_j=j\text{~one
    has~}t_{e_j}\in [\eta (j-1), \eta\cdot j)\right\}.
  \end{equation}
  This is the set of streams $\sigma$ such that collecting $\mathbf{v}$ is
  possible: for every $\sigma\not \in \Gamma$ we have $F = 0$ always.
  Similarly, if $\vb$ does not conform to $\pi$ then $F = 0$.  For any fixed
  $\sigma \in \Gamma$ and $\vb$ that conforms to $\pi$, we have
  \begin{equation}\label{eq:43gh34gh}
    \Pr_r[\Fc]=\prod_{j\in [k]: \pi_j=j} \frac1{|H_j|},
  \end{equation}
  where (as in Algorithm~\ref{alg:walk-from-template}) $H_j$ is the
  set of edges that could could be taken in stage $j$.

  Define
  \begin{equation}
    p := \prod_{j=1}^k \gamma_j = \eta^{|\pi|} \cdot \prod_{j\in [k]: \pi_j=j} \frac1{|H_j|} = \eta^{|\pi|} \Pr_r[\Fc],
  \end{equation}
  where we let $|\pi|:=|\{j\in [k]: \pi_j=j\}|$.  We note that
  $|\pi|\leq k$ for every $\pi$.  This implies that, for any stream
  $\sigma \in \Gamma$ and $\pi$ that $\mathbf{v}$ conforms to,
  \[ 
    \E[r]{\frac{1}{p} F} = \eta^{-|\pi|}.
  \]
  Since $\Pr\brac{\sigma \in \Gamma} = \eta^{|\pi|}$, this means
  \begin{align}\label{eq:Ebound}
    \E[\sigma,r]{ \frac{1}{p} F} = 1.
  \end{align}
  Define $\wh{d}(v_j)$ to be the degree of $v_j$ in
  $\mathbf{v} \cup \sigma[\eta k, 1]$, so $\wh{d}_j = \wh{d}(v_j)$
  when $\Fc$ occurs, 
  and define
  $q := \eta^k\prod_{i=1}^k \frac{1}{\wh{d}(v_{j-1})}$ and
  $q^* := \eta^k\prod_{i=1}^k \frac{1}{d(v_{j-1})}$.  Note that, for any $\vb$,
  $q^*$ depends only on the graph and is independent of the stream order and
  the randomness of the algorithm.  If $\mathbf{v}$ is collected, then it is
  output with probability
  $\alpha = \prod_{j\in [k]} \min\paren*{\frac{\eta }{\gamma_j \wh{d}_{j-1}},
  1}$, which satisfies
  \begin{align}\label{eq:chanceoutput}
    \frac{q^*}{p} \leq \alpha \leq \frac{q}{p}.
  \end{align}
  The upper bound follows from ignoring the $\min(\cdot, 1)$ in the expression,
  and the lower bound follows from the fact that
  $\gamma_j = \frac{\eta}{\abs{H_j}} \geq \frac{\eta}{d(v_{j-1})}$ and $\wh{d}(v_j)
  \le d(v_j)$ for all $j$, so
  \[
    \min\paren*{\frac{\eta }{\gamma_j \wh{d}(v_{j-1})}, 1} \geq
    \min\paren*{\frac{\eta }{\gamma_j d(v_{j-1})}, 1} = \frac{\eta}{\gamma_j
    d(v_{j-1})}.
  \]
  We would like to bound the probability $\mathbf{v}$ is output over
  the streams and internal randomness, which is 
  \begin{align}\label{eq:outputprob}
    \E[\sigma,r]{ \alpha F}.
  \end{align}

  \paragraph{Lower bound.} By~\eqref{eq:chanceoutput} and~\eqref{eq:Ebound},
  \begin{align}
    \E[\sigma,r]{ \alpha F} = \E[\sigma,r]{ \alpha p \frac{1}{p} F} \geq
    \E[\sigma,r]{ q^* \frac{1}{p} F} = q^*
  \end{align}
  as $q^*$ is independent of the stream order and the randomness of the
  algorithm.

\paragraph{Upper bound.}  We have that
\begin{align*}
  \E[\sigma,r]{ \alpha F} \leq \E[\sigma,r]{ \frac{q}{p} F} = \E[\sigma]{q
  \cdot 1_{\sigma \in \Gamma} \cdot \E[r]{\frac{1}{p} F \middle| \sigma}} =
  \E[\sigma]{q \cdot 1_{\sigma \in \Gamma} \cdot \eta^{-|\pi|}} = \E[\sigma]{q
  \middle| \sigma \in \Gamma}.
\end{align*}
So it suffices to show that $q$ is not much bigger than $q^*$ on
average over streams in $\Gamma$.

For any given $i \in \{0, 1, \dotsc, k-1\}$, consider the distribution of the
degree estimate $\wh{d}(v_i)$ over $\sigma \in \Gamma$.  If $v_i$ has $\ell$
distinct incident edges among the walk $\mathbf{v}$, then those edges will
always count in $\wh{d}(v_i)$.  Every other edge will count if and only if its
timestamp is at least $k\eta$, which is an independent binary random variable
with expectation $1-k\eta$.  Let $Y_i$ denote the number of edges that do
\emph{not} count, so $\wh{d}(v_i) = d(v_i) - Y_i$ and $Y_i \sim B(d(v_i) - \ell, k
\eta)$.  We now analyze \[
  \E[\sigma]{\frac{q}{q^*} \middle| \sigma \in \Gamma} = \E[\sigma]{\prod_{i=0}^{k-1} \frac{d(v_i)}{\wh{d}(v_i)}\middle|\sigma \in \Gamma} = \E[\sigma]{\prod_{i=0}^{k-1} \frac{d(v_i)}{d(v_i) - Y_i} \middle| \sigma \in \Gamma}.
\]
Define $\Ic_i$ to be the event that $Y_i \geq d(v_i) / 2$.  By the Bennett
inequality corollary Lemma~\ref{lm:concentration}, if the constant in~\ref{p2}
is chosen to be large enough,
\[
  \Pr[\Ic_i] \leq (3 k \eta)^{d(v_i) / 5}.
\]
On the other hand, when $\Ic_i$ does not occur and $Y_i < d(v_i) / 2$ we
have
\[
  \frac{d(v_i)}{d(v_i) - Y_i}  = \paren*{\sum_{j=0}^\infty (Y_i/d(v_i))^j} \leq   1 + 2 Y_i / d(v_i).
\]
Let $Z_i = 2 Y_i / d(v_i)$ when $\overline{\Ic_i}$ holds and $0$ otherwise.  Since
$\wh{d}(v_i) \geq 1$, we always have $\frac{d(v_i)}{\wh{d}(v_i)} \leq d(v_i)$.
Therefore
\[
  \frac{d(v_i)}{\wh{d}(v_i)} \leq 1 + Z_i + d(v_i) I_i.
\]
where $I_i$ is the indicator for $\Ic_i$. We can now apply Lemma~\ref{lm:prod}
(if the constant in~\ref{p2} is chosen to be large enough) to $(I_0, \dotsc, I_{k-1}),
(Z_0, \dotsc, Z_{k-1})$, and $\wt{\eta} = 3k\eta$,
to say that 
\[
  \E[\sigma]{\frac{q}{q^*} \middle| \sigma \in \Gamma} =
  \E[\sigma]{\prod_{i=0}^{k-1} \frac{d(v_i)}{\wh{d}(v_i)} \middle| \sigma \in \Gamma}
  \leq 1 + 3^k (3k\eta)^{1/5}.
\] 
Therefore the probability we output $\mathbf{v}$ satisfies
\[
  \E[\sigma,r]{ \alpha F} \leq \E[\sigma]{q \middle| \sigma \in \Gamma} = q^*
  \E[\sigma]{\frac{q}{q^*} \middle| \sigma \in \Gamma}  \leq (1 + 3^k (3k \eta)^{1/5})q^*.
\]
For $\eta < 2^{-Ck}$ for sufficiently large $C$, this is at most
$(1 + \eta^{1/7})q^*$ as desired.
\end{proof}

\section{Near-Independence of Constructed Walks}

In this section we prove that the walks constructed by \walksreset($k, \e, s$) (Algorithm~\ref{alg:walks-reset}) are close to independent in total variation distance. We introduce a useful modified version of \walksreset{} and \wft, as well as some notation, before stating the key formal lemmas.

\paragraph{Hybrid \walksreset{} and \wft{} algorithms.}
As a first step we show that the output distribution of \walksreset{} (which
relies on \wft) is close in total variation distance to the output distribution
of a modified version \walksresetH{} (which in turn relies on a modified \wftH)
-- see Algorithm~\ref{alg:walks-reset-mod} and
Algorithm~\ref{alg:walk-from-template-mod} below. 

The main reason behind the introduction of this algorithm is that it will be easier to prove near-independence of several walks generated on the same stream by \walksresetH{} than to perform the same analysis directly on \walksreset.
The proof of closeness of the output of \walksreset{} and \walksresetH{} in distribution proceeds by the `hybrid argument', and to facilitate this argument the procedure \wftH{} takes a parameter $j^*$, and
changes its behavior relative to \wft{} in the first $j^*$
iterations of the outer loop. We note that \wftH{} is
not an actual algorithm that can be run on a stream (it uses information that is
not available as the stream comes in, such as exact vertex degrees and exact
vertex neighborhoods), and is only a useful construct that
facilitates analysis. 
\begin{algorithm}
	\caption{\walksresetH: simulate $s$ samples of $k$-step walk started from uniformly random vertices, with reset (i.e., allowing walks to fail)}  
	\label{alg:walks-reset-mod}

\begin{algorithmic}[1]
\Procedure{\walksresetH($k, \e, s, {\color{red} j^*}$)}{} 
\State \Comment{$k$ is the desired walk length, $\e\in (0, 1)$ is a precision parameter}
\State $\eta \gets \e^{8}\cdot 2^{-Ck}$\Comment{For a large enough constant $C>0$} 
\State $v^i_0\gets $ independent uniform sample from $V$ for $i\in [s]$
\For{$i\in [s]$}
\State Choose $\pi^i\sim \mathcal{U}(\Pi_k)$ \Comment{$\Pi_k$ is the set of walk templates of length $k$, see Definition~\ref{def:template}}
\State $v^i\gets \textsc{WalkFromTemplate{\color{red}Hybrid}}(v^i_0, \pi^i, k, \eta, {\color{red} j^*})$  \Comment{Run in parallel on the same stream}
\EndFor
\State \Return $(v^i)_{i\in [s]}$
\EndProcedure
\end{algorithmic}
\end{algorithm}

\begin{algorithm}
\caption{\walksH: : simulate $\nwalks$ samples of $k$-step walk started from
uniformly random vertices}  \label{alg:simulate-walks-mod}

\begin{algorithmic}[1]
\Procedure{\walksH($k, \e, \nwalks, {\color{red} j^*}$)}{} 
\State \Comment{$k$ is the desired walk length, $\e\in (0, 1)$ is a precision parameter}
\State $\eta \gets \e^{8}\cdot 2^{-Ck}$\Comment{For a large enough constant $C>0$} 
\State $s\gets \nwalks\cdot 100\eta^{-k}\cdot k!$ \Comment{Increase number of
starting nodes to account for failed walks}
\State $v^i_0\gets $ independent uniform sample from $V$ for $i\in [s]$ 
\State $(v^i)_{i\in [s]}\gets \walksresetH(k, \e, s, {\color{red} j^*})$ 
\If{at least $\nwalks$ of the walks $(v^i)_{i\in [s]}$ succeeded}
\State \Return the first $\nwalks$ successful walks
\Else
\State \Return FAIL
\EndIf
\EndProcedure
\end{algorithmic}
\end{algorithm}

We first define
\begin{equation}\label{eq:l-def}
L=\{u\in V: d(u)\leq n^{1/4}\}
\end{equation}
to be the set of `low degree' vertices in the graph (the threshold of $n^{1/4}$ is somewhat arbitrary, but in general this threshold cannot be too high; in particular, it needs to be bounded away from $n$ by at least an $s^2$ factor, where $s$ is the number of samples used in \walksreset). For the first $j^*$ iterations of the main loop \wftH{} (Algorithm~\ref{alg:walk-from-template-mod})
\begin{itemize}
\item if the current vertex $u$ is a high degree vertex (belongs to $V\setminus
L$), samples its next edge uniformly at random from $\delta(u)$, without using
the stream (see line~\ref{line:hd-sample})
\if 0 \item simply outputs FAIL when when taking a back edge (i.e., whenever the
template $\pi$ satisfies $\pi_j<j$) if the current vertex is high degree
(belongs to $V\setminus L$)
\fi
\item when estimating vertex degrees, uses (appropriately scaled) exact degrees
for high degree vertices (see line~\ref{line:deg-est-mod-p-high}), for all $j\leq j^*$. 
\end{itemize}
The changes introduced in \wft{} (Algorithm~\ref{alg:walk-from-template}) with respect to \wftH{} are highlighted in {\color{red} red} in Algorithm~\ref{alg:walk-from-template-mod}.

There is one other difference in the implementations of
\wft{} and \wftH, which
does not affect the output but is convenient for the analysis: in
\wftH, the rejection sampling is done
progressively.  The final acceptance probability $\alpha_k$ for a walk
equals the acceptance probability $\alpha$ in
$\wft$, and the intermediate acceptance
probabilities $\alpha_j$ let us show that the intermediate states are
spread out comparably to a uniform random walk (proved in
Lemma~\ref{lem:walkfromtemplatehybrid0-marginal}).

This means that the chance a random walk is still accepted after $j$ steps is
\[
  \alpha_j := \prod_{g=1}^j \min(\frac{\eta}{\gamma_g \wh{d}_{g-1}^j, 1})
\]

The following lemma shows that the output distribution of $\walksreset(k, \e, s)$ is close to $\walksresetH(k, \e, s, k)$ in total variation
distance: 
\begin{lemma}\label{lm:tvd-high-deg}
  Assuming~\ref{p1}, \ref{p2} and $s\leq n^{1/30}$, for every input
  graph $G$, we have that the output of $\walksreset(k, \e, s)$ is
  $O(s\cdot n^{-1/10})$-close in total variation distance to the
  output of $\walksresetH(k, \e, s, k)$.
\end{lemma}
The proof of Lemma~\ref{lm:tvd-high-deg} is presented in Section~\ref{sec:tvd-high-deg} below.

\begin{algorithm}
	\caption{\wftH: generate a random walk from starting vertex $u_0$ that conforms with a template $\pi$, modified on high degree vertices}  
	\label{alg:walk-from-template-mod}

\begin{algorithmic}[1]
\Procedure{\wftH($v_0, \pi, k, \eta, j^*$)}{} \Comment{$k$ is the desired walk length, $u_0$ is the starting vertex}
\For{$j=1$ to $k$}
\State $u\gets v_{j-1}$
\If{$\pi_j=j$}
\State {\color{red} {\bf If~} $u \in V \setminus L$ and $j\leq j^*$ {~\bf then~}} 
\State {\color{red}{~~~~~~~}$H_j\gets \delta(u)$ }\label{line:hd-sample}
\State {\color{red} {\bf else~}} 
\State ~~~~~~~$H_j\gets $ edges in $\delta(u)\cap \sigma[\eta \cdot (j-1), \eta \cdot j)$\label{line:hj-mod}
\State {\bf If~}{$H_j=\emptyset$} {\bf~then~} \Return FAIL \label{ln:firstfail}
\State $f_j \gets \mathcal{U}(H_j)$ \Comment{Implemented using reservoir sampling}\label{line:sampling-mod}
\State {\color{red} {\bf If~} $u\in V\setminus L$ and $j\leq j^*$ {~\bf then~} $\gamma_j\gets \frac1{d(u)}$ {\bf ~else~}} $\gamma_j\gets \frac1{|H_j|}\cdot \eta$
\Else
%\State {\color{red} {\bf If~}$u\in V\setminus L$ {\bf~then~}\Return FAIL}
\State {\bf If~}{$f_{\pi_j}\not \in \delta(u)$} {\bf~then~} \Return FAIL \label{ln:scndfail}
\State $f_j\gets f_{\pi_j}$
\State $\gamma_j\gets 1$
\EndIf
\State $v_j\gets $ endpoint of $f_j$ other than $u$ \label{line:v-def}
\For{$g\in [j]$} \Comment{Compute degree estimates for vertices on the walk}
\State $\wh{d}^j_{g-1}\gets$ degree of $v_{g-1}$ in $\{f_1, \ldots, f_j\}\cup \sigma[\eta \cdot k, 1]$ \label{line:deg-est-mod}
\State {\color{red} {\bf If~}$v_{g-1}\in V\setminus L$ and $g\leq j^*$ {\bf~then~}}
\State ~~~{\color{red} $\wh{d}^j_{g-1}\gets (1-k\eta) d(v_{g-1})$}\label{line:deg-est-mod-p-high}
%ZZZ
\EndFor
\State $\alpha_j\gets \prod_{g=1}^j \min\{\frac{\eta}{\gamma_g\cdot \wh{d}^j_{g-1}}, 1\}$\Comment{Observe that $\alpha_j \leq \alpha_{j-1}$}
\State \Return FAIL with probability $1-\alpha_j / \alpha_{j-1}$\label{line:rejection-sampling-mod}
\State \Comment{Note that the chance of returning FAIL by step $j$ is $1-\alpha_j$.}
\EndFor
\State \Return $(v_0, \ldots, v_k)$
\EndProcedure
\end{algorithmic}
\end{algorithm}

\subsection{Notation}

A length-$\ell$ ``partial'' walk
$\mathbf{v} = (v_0, \dotsc, v_\ell) \in (V \cup \{\bot\})^{\ell+1}$ consists
of at least 1 vertex of a random walk, followed (possibly) by a series
of $\bot$.  A collection of $s$ partial walks is
$\vv = (\mathbf{v}^i)_{i \in [s]}$.

\begin{definition}[Sampled and collected vertices]\label{def:sampled}
  We say that a partial walk $\mathbf{v}$ is ``collected'' by an
  execution of $\wftH(k, \e, \eta, \ell)$,
  if for each $j$, $v_j$ is the value set by the algorithm or $\bot$
  if the algorithm returns FAIL before setting $v_j$.

  We say that a partial walk $\mathbf{v}$ is ``sampled'' by an
  execution of $\wftH$ in the same
  situation, except that $v_j = \bot$ if the algorithm rejects $v_j$
  by returning FAIL in round $j$.

  For an invocation of \walksresetH$(k, \e, s, \ell)$, we say that
  $\vv$ is collected or sampled if
  $(\mathbf{v}^1, \dotsc, \mathbf{v}^s)$ are sampled or generated by
  the $s$ executions of $\wftH$, respectively.
\end{definition}
In other words, if the walk fails at the end of round $j$, the walk that is
``collected'' still includes $v_j$, while the walk that is sampled does not (if
the walk succeeds, or if the algorithm returns FAIL in line~\ref{ln:firstfail}
or~\ref{ln:scndfail}, they are identical).

\begin{definition}\label{def:psi}
For a length-$\ell$ partial walk $\mathbf{v}$ we define its neighborhood
\begin{equation*}
\begin{split}
\Psi(\mathbf{v})=\bigcup_{0\leq j \leq \ell, v_j\in L} \delta(v_j).
\end{split}
\end{equation*}
For a collection of walks $\vv$ we define
$$
\Psi(\vv)=\bigcup_{\mathbf{v}} \Psi(\mathbf{v}).
$$
\end{definition}

\begin{definition}\label{def:f}
For $j\in {0, 1, \dotsc, k}$ let $\F_j$ denote the following random variables:
\begin{enumerate}
\item The partial walks $\vv_{\leq j-1} := (v^i_{\leq j-1})_{i\in [s]}$ sampled by \walksresetH$(k, \e, s, j)$.
\item Timestamps of edges in $\Psi(\vv_{\leq j-1})$, i.e. 
\begin{equation*}
\begin{split}
\left\{(f, t_f): f\in \Psi(\vv_{\leq j-1})\right\}.
\end{split}
\end{equation*}
\item The internal randomness of $\walksresetH(k, \e, s, j)$ (Algorithm~\ref{alg:walks-reset-mod}) used up to step $j-1$.
\end{enumerate}
We let $\F_0 = \emptyset$.
\end{definition}

\begin{definition}\label{def:ef}
  For a collection of walks $\vv$ and $j\in [k]$ we define
  $E^*_j:=\Psi(\vv_{\leq j -1}) \cup \{(v^i_{j-1}, v^{i'}_{j-1}) \mid
  i, i' \in [s]\}$ to contain all edges out of low degree vertices in
  the first $j$ steps of any walk, combined with all edges between
  vertices visited at position $j-1$ in the different walks.
\end{definition}

\subsection{Proof of Lemma~\ref{lm:tvd-high-deg}} \label{sec:tvd-high-deg}
\begin{lemma}\label{lem:walkfromtemplatehybrid0-marginal}
  Let $v \in V$ uniformly at random and $\pi \sim \mathcal{U}(\Pi_k)$,
  and consider the execution of
  $\wftH(v, \pi, k, \eta, 0)$.  We have for
  any $u \in V$ and $j \in [k]$ that
  \[
    \Pr[v_j = u \mid v_j \neq \bot] \leq 2d(u) / n.
  \]
\end{lemma}
\begin{proof}
  Let $p_v^j(u)$ be the probability that a $j$-step random walk from
  $v$ ends at $u$.  By Claim~\ref{cl:average},
  $\E[v \sim U(V)]{p_v^j(u)} \leq d(u)/n$.
  
  First, consider $j = k$.  The distribution of $v_k$ is the same as
  for the last vertex $u_k$ output by \wft$(v, \pi, k, \eta)$.
  Therefore, per Corollary~\ref{cor:walk-distribution},
  \[
    \Pr[v_k = u] \in [1, 1 + O(\eta^{1/7})] \cdot \lambda \cdot p_v^k(u)
  \]
  for $\lambda = \frac{1}{k!} \eta^k$ independent of $u$ and $v$. Now,
  conditioning on $v_k \not= \bot$ increases this probability by at most a
  factor of $1/\lambda$, as $v_k \not= \bot$ will always hold if the sampled
  template corresponds to a valid walk out of $v_0$ (which happens with
  probability at least $1/k!$) and every edge on that walk is in the right
  length-$\eta$ window (which happens with probability $\eta^k$), and so
  $\Pb{v_k \not= \bot} \ge \lambda$. Therefore,
  \[
    \Pr[v_k = u \mid v_k \neq \bot] \leq (1 + O(\eta^{1/7})) p_v^k(u) \leq 2 d(u)/n.
  \]
  For $j < k$, we note that $v_j$ is distributed essentially the same
  as the last vertex $u_j$ output by \wft$(v, \pi, j, \eta)$.
  (There is one difference: it would be identical if the
  the $\wh{d}$ computed in Line~\ref{line:deg-est} of \wft{} used
  $\sigma[\eta j, 1]$ rather than $\sigma[\eta k, 1]$.  But this
  difference has no bearing on the proof of
  Corollary~\ref{cor:walk-distribution}, which just uses that this
  interval includes $[\eta^{99/100}, 1)$ but not $[0, \eta j]$.)  So we still have that
  $v_j$ is distributed within a $1 +O(\eta^{1/7})$ factor of being
  proportional to a true $j$-step random walk, and the result holds.
\end{proof}

\begin{claim}\label{cl:posterior}
  For $j\in [k]$ the posterior distribution of $(t_f)_{f\in E}$ given
  $\mathcal{F}_j$ is a product distribution.  For every
  $f\in E\setminus E^*_j$ the distribution of $t_f $ is uniform on
  $[0, 1]$.
\end{claim}
\begin{proof}
  The only edge timestamps that influence $\Fc_j$ are those of edges in
  $\Psi(\vv_{\leq j-1})$, and so as we are conditioning on the value of all of
  those, and the prior distribution of  $(t_f)_{f\in E}$ is a product
  distribution of uniform distributions, the result follows.
\end{proof}

\newcommand{\bv}{{\bf v}}

% \begin{observation}\label{o:hd}
% Recalling Definition~\ref{def:ef}, we get for every $u\in V\setminus L$ (recall that the set of low degree vertices $L$ was defined in~\eqref{eq:l-def}) and every realization of $\F_{j-1}$
% \begin{equation}\label{eq:3yht8hg}
% \begin{split}
% \left|\delta(u)\setminus E^*_j \right|&\geq d(u)-s\cdot k\geq (1-\frac{sk}{n^{1/4}}) d(u).
% \end{split}
% \end{equation}
% \end{observation}

\begin{lemma}\label{lm:high-degree-sampling}
  For every $j \in [k]$, consider the execution of
  $\walksresetH(k, \eps, s, j-1)$.  Let $H^i_j$ be the value of $H_j$
  in invocation $i$ of $\wftH$.  For every
  $\F_j$, one has with probability at least $1-n^{-2}$ over timestamps
  of edges in $E\setminus E^*_j$ that, for every $i\in [s]$ such that
  $\pi^i_j = j$ and $v^i_{j-1}\in V\setminus L$, we have that:
  \[
    (1-n^{-1/9}) \eta \cdot d(v^i_{j-1})\leq |H^i_j|\leq (1+n^{-1/9}) \eta \cdot d(v^i_{j-1}),
  \]
  and
  \[
    (1-n^{-1/9}) \eta \cdot d(v^i_{j-1})\leq |H^i_j \setminus E_j^*|\leq (1+n^{-1/9}) \eta \cdot d(v^i_{j-1}).
  \]
  Finally, for every such $i \in [s]$ with
  $v^i_{j-1} \in V \setminus L$ and every $g \in [k]$, the degree estimate
  $\wh{d}^g_{j-1}$ in invocation $i$ satisfies
  \[
    (1-n^{-1/9}) (1-k\eta) \cdot d(v^i_{j-1})\leq \wh{d}(v^j_{j-1}) \leq (1+n^{-1/9}) (1-k\eta) \cdot d(v^i_{j-1}).
  \]
\end{lemma}
\begin{proof}
  For each such $u = v^i_{j-1}$, the edges in
  $S := \delta(u) \setminus E_j^*$ have timestamps that are
  independent and uniform in $[0, 1]$.

  Therefore $|H^i_j \cap S|$ is distributed as the binomial variable
  $B(|S|, \eta)$, so by the Chernoff bound we have with
  $1 - \frac{1}{n^3}$ probability that
  \[
    \big| |H^i_j \cap S| - \eta |S| \big| \leq \sqrt{2\eta |S| \log 2n^3}.
  \]
  Suppose this happens.  Every remaining edge, in
  $E^*_j \cap \delta(u)$, is between $u$ and another vertex collected
  in some walk in $v$; hence there are at most $sk$ such edges.
  
  Then by the triangle inequality:
  \begin{align*}
    \big| |H^i_j| - \eta d(u) \big| &\leq \big| |H^i_j| - |H^i_j \cap S| \big| + \big| |H^i_j \cap S| - \eta |S| \big| + \big| \eta d(u) - \eta |S| \big|\\
                                    &\leq sk + \sqrt{2\eta d(u) \log 2n^3} + \eta sk\\
                                    &\leq \eta d(u) \cdot \left(\frac{2sk}{\eta d(u)} + \sqrt{\frac{2\log 2n^3}{\eta d(u)}}\right)\\
    &\leq n^{-1/9} \eta d(u).
  \end{align*}
  where the last step uses that $d(u) \geq n^{1/4}$,
  $\eta \geq n^{-1/100}$, $s \leq n^{1/30}$, $k \leq \log n$, and $n$
  is sufficiently large.  The bound on
  $|H^i_j \setminus E^*_j| = |H^i_j \cap S|$ is the same, omitting the
  first of the three terms in the triangle inequality.  We then union
  bound over $i \in [s]$.

  We bound $\wh{d}^g_{j-1}$ similarly: let $a$ be the number of edges
  in $S$ that lie in $\sigma[\eta k, 1]$, which is
  $B(|S|, (1-k \eta))$, and so with $1-\frac{1}{n^3}$ probability
  \[
    \big| a - (1 - k \eta)|S| \big| \leq \sqrt{2(1-k\eta) |S| \log
      2n^3}.
  \]
  Then since $\wh{d}^g_{j-1}$ only differs from $a$ on
  $E^*_j \cap \delta(u)$,
  \[
    \big|\wh{d}^g_{j-1} - (1 - k \eta)d(u)\big| \leq \sqrt{2(1-k\eta) d(u) \log
      2n^3} + sk + (1 - k \eta) sk \leq n^{-1/9}(1 - k \eta)d(u)
  \]
  and we again union bound over $i$.
\end{proof}

We first analyze the vertices \emph{collected} in each step, ignoring
the rejection probability.  We then include the rejection probability,
to analyze the vertices \emph{sampled} in each step.

\begin{lemma}\label{lm:high-deg-sample-step-1}
For $j\in [k]$, let $(u_j^i)_{i\in [s]}$ denote the vertices collected by
$\walksresetH(k, \e, s, j-1)$ at step $j$, and $(\wt{u}_j^i)_{i\in [s]}$ denote
the vertices collected by $\walksresetH(k, \e, s, j)$ at step $j$. For
every $\Fc_j$ the total variation distance between $(u_j^i)_{i\in [s]}$ and
$(\wt{u}_j^i)_{i\in [s]}$ conditioned on $\Fc_j$ is bounded by $O(s n^{-1/9})$.

\end{lemma}
\begin{proof}
Define
$$
\I=\{i\in [s]: v^i_{j-1}\in V\setminus L\},
$$
where for every $i\in [s]$, $v^i_{j-1}\in V\cup \{\bot\}$ is the
$(j-1)^{\text{th}}$ vertex collected on the $i^{\text{th}}$ walk, where we let
$v^i_{j-1}=\bot$ if the $i^{\text{th}}$ walk terminated before the $(j-1)^{\text{th}}$
step. Note that $\I\subseteq [s]$, and is quite possibly a proper
subset: $v^i_j$ could be a low degree vertex, and we may also have
$v^i_j=\bot$ for some $i\in [s]$. Also note that
$(v^i_{<j})_{i\in [s]}$, the set of the first $j-1$ vertices traversed
by the constructed walks, is a function of $\mathcal{F}_j$ (see
Definition~\ref{def:f}), and in particular $\I$ also is.

We now modify the sampling of edges incident on high degree vertices
in the invocation of $\walksresetH(k, \e, s, j)$ to avoid $E^*_j$, and
bound the corresponding loss in total variation distance.  Observe
that for every choice of $\F_j$ for every $i\in \I$,
$\abs{\delta(v^i_{j-1})\cap E^*_j} \leq sk$, so the total variation
distance between the uniform distribution over $\delta(v^i_{j-1})$ and
the uniform distribution over $\delta(v^i_{j-1})\setminus E^*_j$ is
bounded by $sk/n^{1/4}$.  Define $(x_{j}^i)_{i\in [s]}$ to match
$\wt{u}_{j}^i$ for $i \notin \I$, and for $i \in \I$ to be sampled from
$\mathcal{U}(\delta(v^i_{j-1})\setminus E^*_j$ as opposed to
$\mathcal{U}(\delta(v^i_{j-1}))$ in line~\ref{line:sampling-mod} in
the $j^{\text{th}}$ step of $\walksresetH(k, \e, s, j)$, so
$TV((\wt{u}_j^i)_{i\in [s]}, (x_{j-1}^i)_{i\in [s]}) \leq O(s^2k \cdot
n^{-1/4})$.

We now perform a similar modification to the edges sampled by high
degree vertices at the $j^{\text{th}}$ step in the invocation of
$\walksresetH(k, \e, s, j-1)$.  Let $\wt{H}^i_j=H^i_j\setminus E^*_j$.

By Lemma~\ref{lm:high-degree-sampling}, with $1 - n^{-2}$ probability
we have both
\begin{equation}\label{eq:834g8g34t8gDF}
(1-n^{-1/9}) \eta \cdot d(v^i_{j-1})\leq |H^i_j|\leq (1+n^{-1/9}) \eta \cdot d(v^i_{j-1}).
\end{equation}
and
\begin{equation}\label{eq:834g8g34t8gDFa}
(1-n^{-1/9}) \eta \cdot d(v^i_{j-1})\leq |\wt{H}^i_j|\leq (1+n^{-1/9}) \eta \cdot d(v^i_{j-1}).
\end{equation}
for every $i \in \I$.

Conditioned on this high probability event $E_H$, we have
using~\eqref{eq:834g8g34t8gDF} and~\eqref{eq:834g8g34t8gDFa} that for every choice of $\F_j$ for
every $i\in \I$ the total variation distance between the uniform
distribution over $H^i_j$ and the uniform distribution over
$\wt{H}^i_j$ is bounded by $O(n^{-1/9})$.  Therefore we can define
$(y_{j}^i)_{i\in [s]}$ to match $u_j^i$ for $i \notin \I$, and for
$i \in \I$ to sample from $\mathcal{U}(\wt{H}^i_j)$ as opposed to
$\mathcal{U}(H^i_j)$ in line~\ref{line:sampling-mod} in the $j^{\text{th}}$
step of $\walksresetH(k, \e, s, j-1)$.  This satisfies
$d_{TV}((u_j^i)_{i\in [s]}, (y_{j}^i)_{i\in [s]}) \leq O( sn^{-1/9})$.

Now, $(x_{j}^i)_{i\in [s]}$ and $(y_{j}^i)_{i\in [s]}$ are
identically distributed conditioned on $\F_j$ and $E_H$.  To see
this, note that $\walksresetH(k, \e, s, j-1)$ and
$\walksresetH(k, \e, s, j)$ behave identically for
$i \notin \I$.  For $i \in \I$,
$x_{j}^i \sim \mathcal{U}(\delta(v_{j-1}^i)\setminus E_j^*)$ and
$y_{j}^i \sim \mathcal{U}(\wt{H}^i_j)$, both independently of the
algorithms' behavior on $i' \neq i$.  Since $\wt{H}^i_j$ is a random
binomial sample of $\delta(v_{j-1}^i)\setminus E_j^*$, $y_{j}^i$ is
also uniform over $\delta(v_{j-1}^i)\setminus E_j^*$ conditioned on
$|\wt{H}^i_j|$, as long as $\wt{H}^i_j \neq \emptyset$; thus it is
conditioned on $E_H$.

As a result, the outputs of the unmodified calls to
$\walksresetH(k, \e, s, j-1)$ and
$\walksresetH(k, \e, s, j)$ are $O(s^2k n^{-1/4} + n^{-2} + s n^{-1/9}) = O(s n^{-1/9})$
close in total variation distance.
\end{proof}

\begin{lemma}\label{lm:tvd-sampled-step-1}
  For $j\in [k]$, if $(v_j^i)_{i\in [s]}$ denote the vertices
  sampled by \walksresetH$(k, \e, s, j-1)$ at step $j$, and
  $(\wt{v}_j^i)_{i\in [s]}$ denote the vertices sampled by
  \walksresetH$(k, \e, s, j)$ at step $j$, then for every $\F_j$ the
  total variation distance between $(v_j^i)_{i\in [s]}$ and
  $(\wt{v}_j^i)_{i\in [s]}$ is bounded by $O(s n^{-1/9})$.
\end{lemma}
\begin{proof}
  By Lemma~\ref{lm:high-deg-sample-step-1}, the distribution of
  vertices \emph{collected} in step $j$ is $O(s n^{-1/9})$ close in
  total variation distance.  Therefore it suffices to show that the
  rejection sampling is similarly close.
  
  We condition on $\F_j$.  Let $(u_j^i)_{i\in [s]}$ denote the
  vertices collected by \walksresetH$(k, \e, s, j-1)$ at step $j$, and
  $(\wt{u}_j^i)_{i\in [s]}$ denote the vertices collected by
  \walksresetH$(k, \e, s, j)$ at step $j$. By
  Lemma~\ref{lm:high-deg-sample-step-1} the total variation distance
  between $(u_j^i)_{i\in [s]}$ and $(\wt{u}_j^i)_{i\in [s]}$ is
  bounded by $O(s n^{-1/9})$. In what follows we analyze the rejection
  sampling step in line~\ref{line:rejection-sampling-mod} of
  \wftH. We show that the fact that the
  degrees of vertices in $V\setminus L$ are estimated using their
  actual degrees (in line~\ref{line:deg-est-mod-p-high} of
  \wftH) as opposed to using the stream (in
  line~\ref{line:deg-est} of \wft) leads to only
  $O(s\cdot n^{-1/9})$ contribution to total variation distance.

  We note that degree estimates computed for vertices in $L$ (i.e.,
  low degree vertices) are the same in both invocations of
  \walksresetH, and consider vertices in $V\setminus L$, i.e. high
  degree vertices. Let $u=u^i_{j-1}$ for some $i\in [s]$, and suppose
  that $u\in V\setminus L$.

  Define
  \[
    \lambda = \prod_{g=1}^{j-1} \frac{\min\paren*{\frac{\eta}{\gamma_g \wh{d}^j_{g-1}}, 1}}{\min\paren*{\frac{\eta}{\gamma_g \wh{d}^{j-1}_{g-1}}, 1}}.
  \]
  Since $\wh{d}^{j-1}_{g-1} \leq \wh{d}^{j}_{g-1}$, $\lambda \leq 1$.
  The probability that $u$ is accepted by the rejection sampling is
  \[
    \frac{\alpha_{j}}{\alpha_{j-1}} = \lambda \min\paren*{\frac{\eta}{\gamma_j
    \wh{d}^j_{j-1}}, 1}.
  \]
  The only difference in the rejection sampling between the two
  executions is that $\gamma_j$ and $\wh{d}^j_{j-1}$ are different.
  Let $\gamma_j$, $\wt{d}^j_{j-1}$ be the values in the $j^* = j-1$
  case, and $\wt{\gamma}_j = \frac{1}{d(u)}$,
  $\wh{\wt{d}}^j_{j-1} = (1-k\eta)d(u)$ be the values these take
  in the $j^* = j$ case.  The total variation incurred by this change
  in rejection sampling probability is thus
  \begin{align*}
    \delta &\leq \abs{ \lambda \min\paren*{\frac{\eta}{\gamma_j \wh{d}^j_{j-1}}, 1} - \lambda \min\paren*{\frac{\eta}{\wt{\gamma}_j \wh{\wt{d}}^j_{j-1}}, 1}}\\
    &\leq \lambda \abs{ \frac{\eta}{\gamma_j \wh{d}^j_{j-1}} - \frac{\eta}{(1-\eta k)}}\\
     &\leq \abs{ \frac{\abs{H_j}}{\wh{d}^j_{j-1}} - \frac{\eta}{(1-\eta k)}}.
  \end{align*}
  Now, by Lemma~\ref{lm:high-degree-sampling}, conditioned on
  $\mathcal{F}_j$ with at least $1 - n^{-2}$ probability we have
  \[
    |H_j| \in (1 \pm n^{-1/9}) \eta d(u)
  \]
  and
  \[
    \wh{d}^j_{j-1} \in (1 \pm n^{-1/9})(1-k\eta) d(u)
  \]
  so that
  \[
    \delta \leq \frac{\eta}{(1-\eta k)} \cdot \left(\frac{1 + n^{-1/9}}{1 - n^{-1/9}} - 1\right) = O(\eta/n^{1/9}).
  \]
  The net result is that every high degree vertex $u$ has a rejection
  sampling step that is $O(n^{-1/9})$ close in the two cases.  This
  means that all $s$ rejection sampling steps are $O(s n^{-1/9})$
  close in the two cases.  Combining with with the collection being
  close (Lemma~\ref{lm:high-deg-sample-step-1}) gives the result.
\end{proof}

\begin{proofof}{Lemma~\ref{lm:tvd-high-deg}}
  We use induction on Lemma~\ref{lm:tvd-sampled-step-1}. The result
  follows by noting that \walksreset$(k, \e, s)$ is the same as
  \walksresetH$(k, \e, s, 0)$ and applying triangle inequality for
  total variation distance, to get a total variation distance of
  $O(sk n^{-1/9}) \leq O(s n^{-1/10})$ for $k = O(\log n)$.
\end{proofof}

\subsection{Near-Independence of Hybrid Algorithm}\label{sec:tvd-low}
To establish the near-independence of our walks, we compare the output
distribution of \walksresetH$(k, \e, s, k)$ run on a random order
stream $\sigma$ to the output of an auxiliary version of
\walksresetH$(k, \e, s, k)$, in which the invocations of
\wftH$(v^i_0, \pi^i, k, \eta, k)$ are run on
independent streams $\wt{\sigma}^i$.

Let $\vv$ denote the $s$ partial walks sampled (see
Definition~\ref{def:sampled}) by
\wftH$(v^i_0, \pi^i, k, \eta, k)$ on stream
$\sigma$, and let $\ww$ denote the $s$ partial walks that would be
sampled if the $s$ invocations each used their own independent random
streams ($\wt{\sigma}_1, \dotsc, \wt{\sigma}_n$).

\begin{definition}
  We say a collection of (partial) walks $\ww$ is ``well-separated''
  if for all $\mathbf{w}, \mathbf{w}' \in \ww$ we have
  $\Psi(\mathbf{w}) \cap \Psi(\mathbf{w}') = \emptyset$.
\end{definition}

Well-separatedness is equivalent to the conditions:
  \begin{itemize}
  \item No vertex in $L$ is visited in more than one walk.
  \item No pair of adjacent vertices in $L$ is visited in more than one walk.
  \end{itemize}

\begin{lemma}\label{lem:w-well-separated}
  $\ww$ is well-separated with $1 - O(k \cdot n^{-1/11})$ probability.
\end{lemma}
\begin{proof}
  By induction on Lemma~\ref{lm:tvd-sampled-step-1}, the distribution
  of walks generated by \walksresetH$(k, \e, 1, k)$ and
  \walksresetH$(k, \e, 1, 0)$ are $O(k n^{-1/9})$ close in total
  variation distance.  Therefore it suffices to show this lemma for
  for $\ww$ drawn from $s$ calls to
  \wftH$(v^i_0, \pi^i, k, \eta, 0)$ on
  independent random streams, i.e. $j^* = 0$ rather than $j^* = k$.

  Let $p_j(u) = \Pr[\mathbf{w}_j^i = u]$ be the marginal probability
  of sampling vertex $u$ in step $j$, which is independent and
  identically distributed over $i \in [s]$.  By
  Lemma~\ref{lem:walkfromtemplatehybrid0-marginal},
  $p_j(u) \leq d(u)/n$ for all $u$ and $j$.

  Consider any $j, j' \in [k]$ and distinct $i, i'\in [s]$.  For any
  fixed $u \in L$, the probability that $w^i_j = u =  w^{i'}_{j'}$ is
  $p_j(u)^2 \leq d(u)^2/n^2 \leq 1/n^{3/2}$.  Taking a union
  bound over $(sk)^2$ choices of $(i,j,i',j')$ and $n$ of $u$, the
  chance that any vertex in $L$ is visited in more than one walk is at
  most $(sk)^2/n^{1/2}$.

  For any fixed edge $(u, v) \in (L \times L) \cap E$, the probability
  that $w^i_j=u$ and $w^{i'}_{j'} = v$ is at most
  $p_j(u) p_j(v) \leq 1/n^{3/2}$.  Taking a union bound over the
  $n \cdot n^{1/4}$ total edges out of vertices in $L$, the chance
  this ever happens is at most $(sk)^2/n^{1/4}$.
\end{proof}

When $\ww$ is well-separated, the edges whose timestamps are looked at
by \wftH$(v^i_0, \pi^i, k, \eta, k)$ are
different across each invocation $i \in [s]$.  This lets us couple the
independent streams case to the dependent streams case.

\begin{lemma}\label{lm:tvd-low} Assuming~\ref{p1}, \ref{p2} and $s\leq n^{1/30}$, for every input graph $G$
$$
d_{TV}\left(\ww, \vv\right)\leq n^{-1/12}
$$
\end{lemma}

\begin{proof} We partition the internal randomness used by \walksresetH{} into
  $s$ disjoint independent random strings $R^1,\ldots, R^s$,
  such that for every $i\in [s]$ the $i^{\text{th}}$ invocation of \wftH{} uses
  string $R^i$.

  Note that the behavior of the $i^{\text{th}}$ invocation of
  \wftH{} is fully determined by its
  randomness $R^i$ and by the timestamps of edges in
  $\Psi(\mathbf{w}^i)$ (the low-degree vertices; see
  Definition~\ref{def:psi}).  Given these timestamps,
  \wftH{} is invariant under other changes to
  the stream.

  Let $\Ev$ denote the event that the independent walk case $\ww$ is
  ``well-separated'', meaning the $\Psi(\mathbf{w}^i)$ do not overlap.
  Per Lemma~\ref{lem:w-well-separated}, $\Ev$ holds with
  $1 - O(kn^{-1/11})$ probability.  When $\Ev$ holds, we can ``couple''
  the independent stream result $\ww$ to the single stream result
  $\vv$ because each element $\vv^i$ depends on disjoint edges in the
  stream. That is, we use the timestamps used by the independent walk
  algorithms to construct a single stream that is both correctly distributed
  (all of its timestamps uniform and independent) such that, if every instance
  of \wftH{} was run on this stream, their executions would be identical to the
  executions of the original independent copies of \wftH{} whenever $\Ev$ holds.

  We now give the details of coupling $\ww$ and $\vv$ under $\Ev$.  We
  construct a distribution $\mathcal{D}$ of a single stream and internal
  randomness $(\sigma, R)$ from $(\wt{\sigma}^i, R^i)_{i\in[s]})$, such that
  the marginal distribution of $(\sigma, R)$ is that of the stream and
  randomness underlying $\vv$, and so that, if the stream and randomness
  underling $\vv$ is set to be $(\sigma, R)$, $\ww = \vv$ whenever $\Ev$ holds.

  Given the independent streams $(\wt{\sigma}^i)_{i\in[s]}$, we can construct
  the dependent instance $(\sigma, R)$ as follows.  We set $R = \wt{R}$, and
  just need to fix the timestamps of edges in $\sigma$.  For each edge $f \in
  E$, we say that an edge $f\in E$ is {\em covered} by $i\in [s]$ if $f\in
  \Psi(\mathbf{w}^i)$. For $f\in E$, let $t_f$ be the timestamp
  of $f$ in $\sigma$ and $\wt{t}_f^i$ be the timestamp of $f$ in
  $\wt{\sigma}^i$.  We set
\begin{equation}
t_f=\left\{
\begin{array}{ll}
\wt{t}^i_f&\text{~if~}i\text{~is the smallest that covers~}f\\
\mathcal{U}[0, 1]&\text{~if $f$ is not covered by any $i\in [s]$.}
\end{array}
\right.
\end{equation}
Note that the resulting distribution of $(t_f)_{f\in E}$ is indeed a
product of uniform distributions over $[0, 1]$. Indeed, one can think
of $(t_f)_{f\in E}$ using the principle of deferred decisions:
starting with $v_0^1$, we run \wftH{} from
$v_0^1$, sampling timestamps of edges as soon as they are needed, and
then proceeding for $v^2_0, v^3_0,\ldots, v^s_0$. The first time the
edge is covered by a walk, its timestamp is sampled from the uniform
distribution, independently of the other timestamps, as required.

It remains to note that conditioned on $\Ev$ for every $f\in E$ there exists at
most one $i\in [s]$ such that $f$ is covered by $i$. Therefore, when $\Ev$
holds, the timestamps of $\Psi(\mathbf{w}^i)$ are the same in $\sigma$ and
$\wt{\sigma}^i$.  Since the randomness is also the same, this means
$\mathbf{v}^i = \mathbf{w}^i$ for each $i$, or $\vv = \ww$.  Since $\Ev$ occurs
with probability $1-kn^{-1/11}$, we have that the two outputs match with
probability $1-kn^{-1/11} > 1 - n^{-1/12}$, as desired. 
\end{proof}

\section{Proof of Theorem~\ref{thm:main}}
We now give 

\begin{proofof}{Theorem~\ref{thm:main}}

\paragraph{Correctness.} Let $S=(v^1_0, \ldots, v^s_0)$, where $v^i_0\sim \mathcal{U}(V)$ are chosen uniformly at random with replacement. We consider four distribution on walks in our proof, which we define below. 

\paragraph{Walks generated by \walksreset{} on independent streams.} Let $\sigma_1,\ldots, \sigma_s$ be independent random order streams. Let $\vv_*=(\mathbf{v}_*^i)_{i\in [s]}$ denote walks generated by \walksreset$(k, \eps, s)$ (Algorithm~\ref{alg:walks-reset}), where the $i^{\text{th}}$ walk is generated on $\sigma^i$. 

\paragraph{Walks generated by \walksresetH{} on independent streams.}
Let $\sigma_1,\ldots, \sigma_s$ be independent random order streams. Let $\vec{\mathbf{x}}=(\mathbf{x}^i)_{i\in [s]}$ denote walks generated by \walksresetH{}$(k, \eps, s, k)$ (Algorithm~\ref{alg:walks-reset-mod}), where the $i^{\text{th}}$ walk is generated on $\sigma^i$.  

\paragraph{Walks generated by~\walksresetH{} on a joint stream.} Let $\vec{\mathbf{y}}=(\mathbf{y}^i)_{i\in [s]}$ denote walks generated by \walksresetH{}$(k, \eps, s, k)$ (Algorithm~\ref{alg:walks-reset-mod}), where every walk is generate on the same stream $\sigma$.

\paragraph{Walks generated by~\walksreset{} on a joint stream.} Let $\vec{\mathbf{v}}=(\mathbf{v}^i)_{i\in [s]}$ denote walks generated by \walksreset{}$(k, \eps, s)$ (Algorithm~\ref{alg:walks-reset}), where every walk is generate on the same stream $\sigma$.

\paragraph{Proof outline.} Note that $\vv$ is the random variable that
we obtain from our random order streaming algorithm~\walksreset. We
first show that it is close in distribution to $\vv_*$, and then show
that the distribution of $\vv_*$ is such that \walks{} has the
desired property.

\paragraph{Step 1: showing that $\vv_* \approx \vv$.}
By Lemma~\ref{lm:tvd-high-deg},
\[
d_{TV}(\vv, \vec{\mathbf{y}})\leq O(s\cdot n^{-1/10})
\]
Second, by Lemma~\ref{lm:tvd-low},
\[
  d_{TV}(\vec{\mathbf{y}}, \vec{\mathbf{x}})\leq n^{-1/12}.
\]
Third, by Lemma~\ref{lm:tvd-high-deg} invoked with $s=1$ applied to
each independent stream, together with triangle inequality for total
variation distance, we have
\[
  d_{TV}(\vec{\mathbf{x}}, \vv_*)\leq \sum_{i=1}^sd_{TV}(\mathbf{x}^i, \mathbf{v}^i_*) \leq  O(s\cdot n^{-1/10}).
\]
Indeed, running \walksreset~(respectively \walksresetH) with every
internal invocation of \wft{} using a separate stream is equivalent to a concatenation of independent instances of \walksreset~(respectively \walksresetH) with $s=1$.  Combining these three equations, the triangle inequality gives
\[
  d_{TV}(\vv, \vv_*) \leq O(s n^{-1/10}) + n^{-1/12} < n^{-1/100}
\]
for sufficiently large $n$, which it will be when $c'$ is sufficiently small,
as $\sqrt{c' \log n} \ge 1$.

\paragraph{Step 2: verifying preconditions of Lemma~\ref{lm:tvd-high-deg} and Lemma~\ref{lm:tvd-low}.}
We now verify that~\ref{p1} and \ref{p2} hold. Let the constants $c, C>0$ be the ones chosen in Lemma~\ref{lm:walk-distribution}.
First,~\ref{p1} is satisfied since 
$$
k\leq \sqrt{c'\log n}\leq c\log n/\log\log n
$$
for a sufficiently small constant $c'>0$ by assumption.
The upper bound in~\ref{p2} is satisfied since
$$
\eta=\e^8\cdot 2^{-Ck}\leq 2^{-Ck}
$$
by the choice of $\eta$ in Algorithm~\ref{alg:walks-reset}.
The lower bound in ~\ref{p2} is satisfied since 
$$
\eta=\e^8\cdot 2^{-Ck}\geq n^{-1/101}\cdot 2^{-Ck}=n^{-1/101-o(1)}\geq n^{-1/100},
$$
since $\e\geq n^{-1/1000}$ and $2^{Ck}=2^{O(\sqrt{\log n})}=n^{o(1)}$ by the assumption on $k$. Furthermore, we have 
$$
 s=\nwalks\cdot (1/\e)^{O(k)}\cdot 2^{O(k^2)}\leq n^{1/100}\cdot 2^{O(c'\log n)}\leq n^{1/30}
$$
since $\nwalks\leq n^{1/100}$ by assumption of the theorem and 
$$
k\leq \min\left\{\frac{c'\log n}{\log (1/\e)}, \sqrt{c'\log n}\right\}
$$
 for a sufficiently small $c'>0$ by assumption of the theorem. 

 \paragraph{Step 3: showing that $\vv_*$ leads to the required distribution.}
 The output of \walks{} (Algorithm~\ref{alg:simulate-walks}) is the
 first $\nwalks$ successful walks from \walksreset{}, which is (per step 1)
 $n^{-1/100}$-close to the first $\nwalks$ successful walks from $\vv_*$,
 for $s = \nwalks \cdot 100 \eta^{-k} \cdot k!$.

 For any fixed $k$-step walk $\mathbf{w}$, let
 \[
   p_*(\mathbf{w}) := \frac{1}{n} \prod_{i=0}^{k-1} \frac{1}{d(w_i)}
 \]
 be the true probability that a $k$-step random walk from a uniform
 vertex equals $\mathbf{w}$.  By taking the average of
 Corollary~\ref{cor:walk-distribution} over uniform initial vertices
 $v$, we have for any $i \in [s]$ that
 \[
   \Pr[\mathbf{v}_*^i = \mathbf{w}] \in [1, 1 + O(\eta^{1/7})] \cdot \frac{\eta^k}{k!} \cdot p_*(\mathbf{w}).
 \]
 As a result,
 \begin{align}\label{eq:vstar-dist}
   \Pr[\mathbf{v}_*^i = \mathbf{w} \mid \mathbf{v}_*^i \text{ succeeds}] \in [1, 1 + O(\eta^{1/7})] \cdot  p_*(\mathbf{w})
 \end{align}
 and
 \[
   \Pr[\mathbf{v}_*^i \text{ succeeds}] \geq \frac{\eta^k}{k!}.
 \]
 By construction, the $\mathbf{v}_*^i$ are independent across
 $i \in [s]$.  We expect at least $\frac{\eta^k}{k!} s = 100\nwalks$ repetitions to succeed, so by a
 Chernoff bound at least $a$ will succeed with at least $1 - 2^{-\nwalks}$
 probability.  Therefore the output of \walks{} is
 $(n^{-1/100} + 2^{-\nwalks})$-close to the distribution of $a$ independent
 samples from
 $(\mathbf{v}_*^i = \mathbf{w} \mid \mathbf{v}_*^i \text{ succeeds})$.
 By~\eqref{eq:vstar-dist}, each such sample is
 $O(\eta^{1/7}) < \eps$-close to a uniform random walk.
\paragraph{Step 4: space complexity.} The expected space complexity of a single
invocation of \wft{} (Algorithm~\ref{alg:walk-from-template}) is $O(k)$, so the
overall space complexity is $(1/\e)^{O(k)} 2^{O(k^2)}\nwalks$, as required.
\end{proofof}

%% file: digraph-lb.tex
\section{Digraph Lower Bound}
\label{sec:dglb}
In this section, we prove that both sampling random walks and approximating the
PageRank of a given vertex set are hard in \emph{directed} graph streams.
\dgrwlb*
\dgprlb*
We will prove these by a reduction from $\ind_n$. In this one-way
communication problem Alice has a string $x \in \bool^n$ while Bob has an index
$I \in \brac{n}$. Alice must send Bob a message such that he can determine
$x_I$. We will show that a \emph{uniform} instance of this problem can be
converted into a random graph stream such that approximating the random walk
distribution or the PageRank vector allows solving indexing.

The following is a well-known consequence of information theory---for
completeness, we include a proof in Appendix~\ref{app:indlb}.
\begin{restatable}{lemma}{indistlb}
\label{lm:indistlb}
Let $\varepsilon > 0$ be any constant. Any protocol that solves $\ind_n$ on a
uniform input with probability $1/2 + \varepsilon$ requires $\bOm{n}$
communication in expectation.
\end{restatable}

\subsection{Graph Distribution}
\label{sec:harddigraphdist}
We start by defining a distribution on length-$\bO{n}$ directed graph streams
that Alice and Bob can construct (with a prefix belonging to Alice and a
postfix belonging to Bob, and interleaved edges that they construct using
shared randomness that are independent of their input) using their inputs to
$\ind_n$.

We will show that the graph streams correspond to randomly choosing a graph and
then uniformly permuting its edges, and any algorithm that generates a
distribution that is $\varepsilon$-close to either the random walk
distribution or the PageRank distribution for some constant $\varepsilon < 1/4$
will be able to use this to solve $\ind_n$.

\paragraph{Vertices.} There will be $\beta n$ vertices $(a_i)_{i=1}^{\beta n}$
where $\beta \in \Nbb_{>0}$ is a constant depending on $\varepsilon$, a single
vertex $b$, $n$ vertices $(c_i)_{i=1}^n$, and two pairs of vertices $\set{d_0,
e_0}$, $\set{d_1, e_1}$. Every vertex in the graph will have a path to either
$\set{d_0,e_0}$ or $\set{d_1,e_1}$, which will function as ``sink'' sets.

\paragraph{Fixed Edges.} These edges will not depend on either player's input.
The players can use shared randomness to insert them uniformly at random among
their other edges.

They are $(a_ib)_{i=1}^{\beta n}$, i.e.\ a star of $\beta n$ edges pointing
into $b$ and the four edges $d_0e_0$, $e_0d_0$, $d_1e_1$, and $e_1d_1$ (i.e.
meaning that each of $\set{d_0, e_0}$ and $\set{d_1, e_1}$ is a 2-vertex loop).

\paragraph{Alice's edges.} Let $\pi$ be a uniformly random permutation of
$\brac{n}$, and let $J$ be drawn uniformly from $\set{0, \dots, n}$.  These
will be used to define the boundary between Bob's edges and Alice's edges.

Recall that Alice's input is a string $x \in \bool^n$. For each $i \in
\brac{J}$, Alice has the edge $c_{\pi(i)}d_{x_{\pi(i)}}$, with
$c_{\pi(1)}d_{x_{\pi(1)}}$ first, $c_{\pi(2)}d_{x_{\pi(2)}}$ second, and so on.

\paragraph{Bob's edges} Recall that Bob has the index $I$. His first edge will
be $bc_I$. Then, for each $i \in \set{J+1, \dots n}$, he has the edge
$c_{\pi(i)}d_{y_{\pi(i)}}$, where $y$ is a random $n$-bit string.

\begin{lemma}
\label{lm:digraphisrand}
The graph stream described above is a uniformly random order graph stream.
\end{lemma}
\begin{proof}
Fix any value of Bob's index $I$. If we were to randomly draw a string $z \in
\bool^n$, randomly order edges $(c_id_{z_i})_{i=1}^n$, and then insert $bc_I$
and then our fixed edges randomly in this stream, this would give a uniformly
random order graph stream (corresponding to drawing a graph from a fixed
distribution and then randomly permuting its edges).

But this would also be identically distributed to our graph stream, as the
strings $x$ and $y$ are both drawn uniformly at random from $\bool^n$. So we
have a uniformly random graph stream.
\end{proof}

\begin{lemma}
\label{lm:digraphcorr}
With probability $3/4$, every vertex in $(a_i)_{i=1}^{\beta n}$ has a
length-$3$ path from it to $\set{d_{x_I},e_{x_I}}$, and (up to prefixes) that
is the only
path out of that vertex.
\end{lemma}
\begin{proof}
Each has an edge to $b$, which has a single edge to $c_J$. With probability
$1/2$, $\pi(I) \le J$, and so the unique edge out of $c_J$ points to $d_{x_I}$.
Otherwise, it points to $d_{y_I}$, which is $d_{x_I}$ with probability $1/2$,
as $y$ is drawn at random.
\end{proof}

We will now prove that both the random walk distribution and the PageRank
vector on this graph can be used to determine the answer to $\ind_n$. Note that
there are no vertices with out-degree $0$ in this graph, so we do not need to
concern ourselves with what a random walk or the PageRank walk should do
when encountering such a vertex.
\subsection{Random Walks}
\begin{lemma}
\label{lm:rwtoprot}
Let $\Ac$ be any algorithm using $S$ space such that, when given a sample from
the distribution on graph streams above, the output of $\Ac$ is
$\varepsilon$-close in total variation distance to sampling a $k$-step random
walk from a random vertex in the graph given by the stream, where $k \ge 3$.
Then there is a protocol for $\ind_n$ on uniform inputs that uses $S$ space and
succeeds with probability $3/4 - \varepsilon - 6/\beta$.
\end{lemma}
\begin{proof}
The protocol will be as follows:
\begin{enumerate}
\item Alice and Bob use their $\ind$ input to construct a corresponding random
graph stream.
\item Using $S$ bits of communication from Alice to Bob, they run $\Ac$ on the
stream.
\item If the walk output by $\Ac$ ends in $\set{d_z,e_z}$ for $z \in \bool$,
Bob answers $z$. Otherwise he answers arbitrarily.
\end{enumerate}
Now, by Lemma~\ref{lm:digraphcorr}, with probability $3/4$ any length-$k$
random walk starting from a vertex in $(a_i)_{i=1}^{\beta n}$ will reach
$\set{d_{x_I}, e_{x_I}}$, and as that set has no out-edges it will stay there.
The probability of starting at one of these vertices is
\begin{align*}
\frac{\beta n}{\beta n + 1 + n + 4} &= \frac{1}{1 + 1/\beta + 5/n\beta}\\
&= \frac{1}{1 + 6/\beta}\\
&\le 1 - 6/\beta
\end{align*}
and so the probability that the distribution given by $\Ac$ ends in
$\set{d_{x_I}, e_{x_I}}$ is at least \[
3/4 - \varepsilon - 6/\beta
\]
proving the correctness of the protocol. The construction of the stream itself
is done entirely with public randomness, so this gives an $S$ bit public
randomness protocol. As we are working with a fixed input distribution, this
means that there is also an $S$ bit private randomness protocol, by fixing
whichever set of public random bits maximizes the probability of success over
the uniform distribution.
\end{proof}

We are now ready to prove Theorem~\ref{thm:dgrwlb}.
\dgrwlb*
\begin{proof}
Set $\beta = \frac{12}{1/4 - \varepsilon}$. Then, for each $n$, we can
construct the family given by the stream distribution described in
section~\ref{sec:harddigraphdist} from $\ind_{n'}$, where $n' = \bOm{n}$ while
keeping the total number of vertices and edges below $n$. By
Lemma~\ref{lm:rwtoprot}, any algorithm that $\varepsilon$-approximates the
distribution of length-$k$ random walks on graphs drawn from this family gives
a protocol for $\ind_{n'}$ that succeeds with probability \[
3/4 - \varepsilon - 6/\beta = 1/2 + \frac{1/4 - \varepsilon}{2}
\]
and so by Lemma~\ref{lm:indistlb} it uses $\bOm{n'} = \bOm{n}$ space, 
where the constant depends only on $\varepsilon$ and $\beta$. So as $\beta$
depends only on $\varepsilon$ the result follows.
\end{proof}
\subsection{PageRank}
\begin{lemma}
\label{lm:prweightcr}
With probability at least $3/4$ over the construction of the graph stream, the
PageRank vector with reset probability $\alpha$ has support at least $(1 -
6/\beta)(1 - \alpha)^3$ on $\set{d_{x_I},e_{x_I}}$.
\end{lemma}
\begin{proof}
Recall that the PageRank vector is the stationary distribution of the Markov
chain in which each step is a random walk step with probability $(1 - \alpha)$
and a jump to a uniformly random vertex with probability $\alpha$. By
Lemma~\ref{lm:digraphcorr}, with probability $3/4$ over the graph generation
process, every length-$3$ or greater walk from a vertex in $(a_i)_{i=1}^{\beta
n}$ reaches $\set{d_{x_I},e_{x_I}}$, and as this set has no out-edges, it stays
there.

Suppose this holds. Then starting from any point, after $k \ge 3$ steps on the chain, 
sufficient criteria to be in $\set{d_{x_I}, e_{x_I}}$ are that
\begin{enumerate}[label=\textbf{(\arabic*)}]
\item There has been at least one jump, but not in the last $3$ steps. 
\label{en:prjumphap}
\item That jump went to a vector in $(a_i)_{i=1}^{\beta n}$.
\label{en:prjumpright}
\end{enumerate}
\noindent
Criterion~\ref{en:prjumphap} will hold with probability at least \[
(1 - (1 - \alpha)^{k-3})(1 - \alpha)^3
\]
which converges to $(1 - \alpha)^3$ as $k \rightarrow \infty$.
Conditioned on this, criterion~\ref{en:prjumpright} will hold with probability
$\frac{\beta n}{ \beta n + 1 + n + 2} \ge 1 - 6/\beta$. So as $k \rightarrow
\infty$, the probability that a walk on the chain is in $\set{d_{x_I},
e_{x_I}}$ converges to at least
\[
(1 - 6/\beta)(1 - \alpha)^3
\] completing the proof.
\end{proof}
\begin{lemma}
\label{lm:prtoprot}
Let $\Ac$ be any algorithm using $S$ space such that, when given a sample from
the distribution on graph streams above, it with probability $1 - \delta$
outputs an $\varepsilon$-additive approximation to the PageRank of the set
$\set{d_0,e_0}$ with reset probability $\alpha$, where $\varepsilon < (1 -
6/\beta)(1 - \alpha)^3 - 1/2$.  Then there is a protocol for $\ind_n$ on
uniform inputs that uses $S$ space and succeeds with probability $3/4 -
\delta$.
\end{lemma}
\begin{proof}
The protocol will be as follows:
\begin{enumerate}
\item Alice and Bob use their $\ind$ input to construct a corresponding random
graph stream.
\item Using $S$ bits of communication from Alice to Bob, they run $\Ac$ on the
stream, estimating the PageRank of $\set{d_0,e_0}$.
\item If the sample output by $\Ac$ is at least $1/2$, Bob outputs $0$.
Otherwise he outputs $1$.
\end{enumerate}
Now, by Lemma~\ref{lm:prweightcr}, with probability $3/4$, the PageRank vector
of the graph has support at least \[
(1 - 6/\beta)(1 - \alpha)^3
\]
on $\set{d_{x_I},e_{x_I}}$ and so with probability $3/4 - \delta$, the algorithm
will report the correct answer.

The construction of the stream itself is done entirely with public randomness,
so this gives an $S$ bit public randomness protocol. As we are working with a
fixed input distribution, this means that there is also an $S$ bit private
randomness protocol, by fixing whichever set of public random bits maximizes
the probability of success over the uniform distribution.
\end{proof}

We are now ready to prove Theorem~\ref{thm:dgprlb}.
\dgprlb*
\begin{proof}
Set $\beta = \frac{12}{(1 - \alpha)^3 - \frac{1}{2} - \varepsilon}$.  Then, for
each $n$, we can construct the family given by the stream distribution
described in section~\ref{sec:harddigraphdist} from $\ind_{n'}$, where $n' =
\bOm{n}$ while keeping the total number of vertices and edges below $n$. By
Lemma~\ref{lm:prtoprot}, as \[
(1 - 6/\beta)(1 - \alpha)^3 - 1/2 \ge (1 - \alpha)^3 - 6/\beta - 1/2  \ge
\frac{1}{2}((1 - \alpha)^3 - 1/2 + \varepsilon) > \varepsilon
\]
this gives a protocol for $\ind_{n'}$ that succeeds with probability \[
3/4 - \delta > 1/2
\]
and so by Lemma~\ref{lm:indistlb} it uses $\bOm{n'} = \bOm{n}$ space, 
where the constant depends only on $\alpha$, $\varepsilon$, and $\beta$. So as $\beta$
depends only on $\alpha$ and $\varepsilon$ the result follows.
\end{proof}

%% file: app-technical-claims.tex
\section{Omitted Proofs of Technical Results}
\label{app:technicalproofs}
\average*
\begin{proof}
The base case is given by $k=0$: we have $\expect_{v\sim \mathcal{U}(V)} [p^k_v(u)]=1/n\leq d(u)/n$. 
The inductive step is given by:
\begin{equation*}
\begin{split}
\E[v\sim \Uc(V)] {p^{k+1}_v(u)}&= \sum_{w \in \delta(u)} \frac{1}{d(w)} \E[v \sim
\Uc(v)]{p_v^k(w)}\\
&\leq \sum_{w \in \delta(u)} \frac{1}{d(w)} \cdot  \frac{d(w)}{n}\\
&= \frac{d(u)}{n}.
\end{split}
\end{equation*}
\end{proof}

We will require Bennett's inequality.
\begin{theorem}[Bennett's inequality, Theorem~2.9 in~\cite{DBLP:books/daglib/0035704}]\label{thm:bennett}
Let $X_1,\ldots, X_n$ be independent random variables with finite variance such that $X_i\leq b$ for some $b>0$ almost surely for all $i\in [n]$. Let 
$$
S=\sum_{i\in [n]}\paren*{X_i-\expect[X_i]}
$$
and let $v=\sum_{i\in [n]} \expect[X_i^2]$. Then for any $t>0$ 
$$
\Pr[S\geq t]\leq \exp\left(-\frac{v}{b^2}h\left(\frac{bt}{v}\right)\right),
$$
where $h(u)=(1+u)\ln (1+u)-u$ for $u>0$.
\end{theorem}

In applying Bennett's inequality, we will need the following two properties of $h$:
\begin{claim}\label{cl:h-mon}
The function $v\cdot h\left(\frac{1}{v}\right)$ is monotone decreasing in $v$, where $h(u)=(1+u)\ln (1+u)-u$ for $u>0$.
\end{claim}
\begin{proof}
\begin{equation*}
\begin{split}
\frac{d}{dv}\left(v\cdot h\left(\frac{1}{v}\right)\right)&=\frac{d}{dv}\left(v\cdot \left[\left(1+\frac{1}{v}\right)\ln \left(1+\frac{1}{v}\right)-\frac{1}{v}\right]\right)\\
&=\frac{d}{dv}\left((v+1)\ln \left(1+\frac{1}{v}\right)-1\right)\\
&=\frac{d}{dv}\left((v+1)(\ln (v+1)-\ln v)-1\right)\\
&=\ln \left(v+1\right)-\ln v + 1-1-1/v\\
&=\ln \left(1+1/v\right)-1/v\\
&\leq 0
\end{split}
\end{equation*}
since $\ln (1+x)\leq x$ for all $x\in (-1, +\infty)$.
\end{proof}

\begin{claim}\label{cl:h-lb}
For every $u\geq 0$ one has $h(u)=(1+u)\ln (1+u)-u\geq \frac{1}{2} u\ln u$.
\end{claim}
\begin{proof}
When $u = 0$ both sides are equal, so it will suffice to verify that
$$
\frac{d}{du}\left((1+u)\ln (1+u)-u-\frac{1}{2} u\ln u\right)=-\frac1{2}-\frac1{2}\ln u+\ln(1+u),
$$ 
is nonnegative for all $u\geq 0$. The latter claim can be verified by observing that the function on the rhs above goes to $+\infty$ as $u\to 0+$ and as $u\to +\infty$, and its derivative
$$
\frac{d}{du}\left(-\frac1{2}-\frac1{2}\ln u+\ln(1+u)\right)=-\frac1{2u}+\frac1{1+u}
$$
has exactly one zero at $u=1$, where $h(u)$ is positive:
$\left.-\frac{1}{2}-\frac{1}{2}\ln u+\ln(1+u)\right|_{u=1}=-\frac{1}{2}+\ln 2>0$.
\end{proof}

\concentration*
\begin{proof}
We have $\alpha_i X_i\leq 1$ with probability $1$, and 
$$
v=\sum_i \expect[(\alpha_i X_i)^2]\leq  \sum_i \expect[\alpha_i X_i]=\eta d.
$$
Noting that the function $v h(\frac{x}{v})$ is monotone decreasing in $v$ for
any $x \geq 0$ (by applying Claim~\ref{cl:h-mon} after rescaling $v$), we get,
letting $t=(1/2-\eta) d$ in Bennett's inequality 
(Theorem~\ref{thm:bennett}) with $b=1$, $t = (1/2-\eta)d$, and $v$,
\begin{equation*}
\begin{split}
\Pr\left[\sum_i \alpha_i X_i\geq d/2\right]&= \Pr\left[\sum_i \alpha_i
X_i-\expect[\alpha_i X_i]\geq (1/2-\eta)d\right]\\
&\leq \exp\paren*{-v\cdot h\paren*{\frac{(1/2-\eta) d
}{v}}}\\
&\leq \exp\paren*{-\eta d\cdot h\paren*{\frac{(1/2-\eta) d
}{\eta d}}}\\
&= \exp\left(-\eta d\cdot h\left(\frac{1/2-\eta}{\eta}\right)\right)\\
\end{split}
\end{equation*}
Noting that $h(u)\geq \frac{1}{2}u\ln u$ for all $u\geq 0$ by Claim~\ref{cl:h-lb}, we get
\begin{equation*}
\begin{split}
\Pr\left[Y\geq d/2\right]&\leq \exp\left(-\frac{1}{2}\eta d\cdot
\frac{1/2-\eta}{\eta}\cdot \ln\left((1/2-\eta) /\eta\right)\right)\\
&\leq \exp\left(-(d/5) \cdot \ln(1/3\eta)\right)\\
&\leq (3\eta)^{d/5}\\
\end{split}
\end{equation*}

\end{proof}

\productlm*
\begin{proof}
  We expand the product and show that each monomial is small in
  expectation.  There are $3^k$ terms, each of the form
  \[
    \expect\left[\prod_{i \in S_1} Z_i \prod_{j \in S_2} q_j E_j \right]
  \]
  for disjoint sets $S_1, S_2 \subseteq [k]$. First, $S_1 = S_2 = \emptyset$ corresponds to the leading $1$ term. Next, when $S_1 \neq \emptyset$ but $S_2 = \emptyset$, we have for any
  $i^* \in S_1$ that
  \[
    \expect\left[Z_{i^*} \cdot \prod_{i \in S_1 \setminus \{i^*\}} Z_i\right] \leq \expect[Z_{i^*}] = \wt{\eta}.
  \]
  Finally, if $S_2 \neq \emptyset$, let $j^* \in S_2$ be of maximal
  $q_{j^*}$, and define $q = q_{j^*}$.  Then
  \[
    \expect\left[\prod_{i \in S_1} Z_i \prod_{j \in S_2} q_j E_j \right] \leq
    \left(\prod_{j \in S_2} q_j\right) \Pr[E_{j^*}] \leq q^k \cdot
    \wt{\eta}^{q/5}
  \]
We now upper bound the last term on the rhs. 
The function
\begin{align*}
  f(q) =  q^k\cdot \wt{\eta}^{q/5}
\end{align*}
has, for all positive integers $q$,
\begin{align*}
  f(q+1)/f(q) = (1 + 1/q)^k \wt{\eta}^{1/5} \leq 2^k \wt{\eta}^{1/5} < (e^{5k} \wt{\eta})^{1/5} \leq 1.
\end{align*}
Therefore $f(q)$ is maximized (over positive integers) when $q = 1$ and 
$f(q) = \wt{\eta}^{1/5}$.  Therefore every term other than the leading $1$ is
at most $\wt{\eta}^{1/5}$.  There are $3^k$ terms, giving the result.
\end{proof}

%% file: app-timestamps.tex
\section{Generating Timestamps in the Stream}
\label{app:streamstamps}
In this section we show how an algorithm can generate $n$ timestamps in a
streaming manner, corresponding to drawing $n$ uniform random variables from
$(0,1)$ and then presenting each in order with $\poly(1/n)$ precision, using
$\bO{\log n}$ bits of space.

Let $(\Xb_i)_{i=1}^n$ denote $n$ variables drawn independently from
$\Uc(0,1)$ and then ordered so that $\Xb_i \le \Xb_{i+1}$ for all $i \in
\brac{n-1}$. By standard results on the order statistics (see
e.g.\ page 17 of~\cite{DN03}), the distribution of $(\Xb_i)_{i=j+1}^n$ depends
only on $\Xb_j$, and in particular they are distributed as drawing $(n-j)$
samples from $(\Xb_j, 1)$.

So then, to generate $(\Xb_i)_{i=1}^n$ with $\poly(1/n)$ precision in the
stream it will suffice to, at each step $i+1$, use $\Xb_{i}$ to generate
$\Xb_{i+1}$ (as sampling from the minimum of $k$ random variables to
$\poly(1/n)$ precision can be done in $\bO{\log n}$ space). We will only need
to store one previous variable at a time, to $\poly(1/n)$ precision, and so
this algorithm will require only $\bO{\log n}$ space.

%% file: pagerank.tex
%!TEX root = ./main.tex
\section{Proofs of Corollary~\ref{cor:rp} and Corollary~\ref{cor:pagerank}}\label{app:pagerank}

\begin{proofof}{Corollary~\ref{cor:rp}} 
 By Theorem~\ref{thm:main} the output of \walksreset{} is $2^{-\Omega(D)}$-close in total variation distance to $a$ samples of $\e$-approximate samples of the $k$-step walk in $G$. We first analyze the algorithm assuming full independence of $(v^i)_{i\in [a]}$. Under this assumption we have
$$
(1-O(\e))\cdot \text{rp}(G)\leq \expect_{\mathbf{v}^i}[v^i_k=v^i_0]\leq (1+O(\e))\cdot \text{rp}(G), 
$$
and get by Chebyshev's inequality, using that $a=\frac{D}{\e^2}$ for a sufficiently large constant $D>0$,  
$$
\Pr\left[\left|\frac1{a}\sum_{i\in [a]} \mathbb{I}[v^i_k=v^i_0]-\text{rp}(G)\right|>\e\right]\leq 99/100.
$$
Since the true $(v^i)_{i\in [a]}$ are $2^{-\Omega(D)}$-close to independent, the success probability is lower bounded by $9/10$ as long as $D$ is larger than an absolute constant, as required.
\end{proofof}

\begin{proofof}{Corollary~\ref{cor:pagerank}}
We assume that the walks output by the invocation of \approxpagerank (Algorithm~\ref{alg:approx-pagerank}) are independent, and account for the $2^{-\Omega(D)}$ additional term in the failure probability due to the total variation distance to independence at the end of the proof.

Define the truncated PageRank vector $\bar p_\alpha$ by 
\begin{equation*}
\begin{split}
\bar p_{\alpha}&=\sum_{k=0}^{\lceil \frac{2}{\alpha}\log 1/\e\rceil} \alpha (1-\alpha)^k M^k\cdot \frac{\mathbbm{1}}{n},
\end{split}
\end{equation*}
and note that under the full independence assumption of the walks we get that the estimator $\wh{p}$ computed by Algorithm~\ref{alg:approx-pagerank} satisfies
$$
\expect[\wh{p}]=\bar p_\alpha(T).
$$

We thus get by Chebyshev's inequality, using that $a=\frac{D}{\e^2}$ for a sufficiently large constant $D>0$,  
$$
\Pr\left[\left|\wh{p}-\bar p_\alpha(T)\right|>\e/2\right]\leq 99/100.
$$

We now note that
\begin{equation*}
\begin{split}
|p_\alpha(T)-\bar p_\alpha(T)|&\leq \|p_\alpha-\bar p_\alpha\|_1\\
&=\left\|\sum_{k>\frac{2}{\alpha}\log 1/\e} \alpha (1-\alpha)^k M^k\cdot \frac{\mathbbm{1}}{n}\right\|_1 \\
&\leq (1-\alpha)^{\frac{2}{\alpha}\log 1/\e}\\
&\leq \e/2\\
\end{split}
\end{equation*}
since $\e<1/2$ by assumption of the corollary. 

It remains to note that since the walks output by the invocation of \approxpagerank are $2^{-\Omega(D)}$-close to independent, we have by combining the above two bounds that $|\wh{p}-p_\alpha(T)|\leq \e$ with probability at least $99/100-2^{-\Omega(D)}\geq 9/10$ as long as $D$ is larger than an absolute constant, as required.

Finally, we verify that the preconditions of Theorem~\ref{thm:main} are satisfied. We invoke Theorem~\ref{thm:main} with $k=\left\lceil\frac{2}{\alpha} \log(1/\e)\right\rceil$. Letting $c'$ be the constant from Theorem~\ref{thm:main}, we get, using the assumption that $\e=2^{-o(\sqrt{\log n})}$, that 
$$
\min\left\{\frac{c'\log n}{\log (1/\e)}, \sqrt{c'\log n}\right\}=\sqrt{c'\log n}.
$$
Thus,  using the assumption of the corollary that $\frac1{\alpha}\leq \frac{\sqrt{\log n}}{4\log (1/\e)}$, we get
$$
k=\left\lceil \frac{2}{\alpha} \log(1/\e)\right\rceil\leq \sqrt{c'\log n}, 
$$
as required.
\end{proofof}

%% file: app-index.tex
\section{Proof of Lemma~\ref{lm:indistlb}}
\label{app:indlb}
In this section we give a proof of Lemma~\ref{lm:indistlb} through a standard
information-theoretic argument. We will need the following definitions and
results from information theory (see, e.g.~\cite{CT91}).

For random variables $X$, $Y$, $Z$, with $p_x = \Pb{X = x}$, define
\begin{description}[font=\normalfont\itshape]
\item[Entropy] $H(X) = \E[X]{-\log p_X}$. For $p \in \brac{0,1}$ we write
$H(p)$ for the entropy of a Bernoulli random variable with parameter $p$.
\item[Conditional Entropy] $H(X|Y) = \E[Y]{\E[X]{-\log p_X | Y}}$.
\item[Mutual Information] $I(X;Y) = H(X) - H(X|Y) = H(Y) - H(Y|X)$
\item[Conditional Mutual Information] $I(X;Y|Z) = H(X|Z) - H(X|Y,Z)$
\end{description}
We will use the fact that $H(p)$ is concave and uniquely maximized at $p =
1/2$, the fact that $I((X_i)_{i=1}^n; Y) = \sum_{i=1}^nI(X_i;Y|(X_j)_{j=1}^i)$,
and Fano's data processing inequality.
\begin{theorem}[Fano's inequality]
Let $X$, $Y$ be discrete random variables, and $P_e = \Pb{g(Y) \not= X}$, where
$g$ is some function on the support of $Y$ and $X$ is supported on
$\mathscr{X}$. Then \[
H(P_e) + P_e(\log\abs{\mathscr{X}} - 1) \ge H(X|Y)
\]
\end{theorem}

\indistlb*
\begin{proof}
Suppose we had such a protocol. Then for each $i \in \brac{n}$, with $X$ as
$x_i$, the $i^\text{th}$ bit of Alice's input, and $Y$ as the message $M$ she
sends, and $P_e$ as the probability that Bob gives the wrong answer when $I = i$, we have \[
H(P_e) \ge H(X|Y)
\]
as $X_i$ is supported on two elements. So for each $i \in \brac{n}$, if $p_i$
is the probability Bob \emph{succeeds} when $I = i$, we have \[
H(x_i | M) \le H(1 - p_i) = H(p_i)
\]

Then, as the protocol succeeds with probability $1/2 + \varepsilon$, and Bob's
index $I$ is uniform on $\brac{n}$, $1/2 + \varepsilon = \sum_{i=1}^n p_i$, so
\begin{align*}
I(x; M) &= \sum_{i=1}^n I(x_i; M | (x_j)_{j=1}^i)\\
&= \sum_{i=1}^n \paren*{H(x_i | (x_j)_{j=1}^i) - H(x_i | M, (x_j)_{j=1}^i))}\\
&= \sum_{i=1}^n \paren*{1 - H(x_i | M, (x_j)_{j=1}^i)} & \mbox{as the $x_i$ are independent}\\
&\ge n -  \sum_{i=1}^nH(x_i | M)\\
&\ge n - \sum_{i=1}^nH(p_i)\\
&\ge n\paren*{1 - H\paren*{\sum_{i=1}^np_i}} & \mbox{by concavity}\\
&= n(1 - H(1/2 + \varepsilon))\\
&= \bOm{n}
\end{align*}
as $\varepsilon$ is a non-zero constant. So in particular $H(M) = \bOm{n}$ and
so $M$ is at least $n$ bits on average.

\end{proof}

%% file: app-query-lb.tex
\section{Lower Bound for Chosen Vertices}
\label{app:chosenvertexlb}
In this appendix, we prove that, if instead of sampling a walk from a randomly
chosen vertex we need to sample a walk from a chosen vertex, any algorithm that
gives better than a $1/2 - 2^{-2-\floor{\frac{k}{2}}}$-approximation to the
distribution of length $k\ge2$ random walks on \emph{undirected} graphs needs
at least $\bOm{n}$ space.
\begin{theorem}
\label{thm:chosenvertexlb}
For any constant $\varepsilon < 1/4 - 2^{-2-\floor{\frac{k}{2}}}$, the following
holds for all $k \ge 2$ and all $n$: there is a family of (undirected) graphs
with $n$ vertices and edges such that any algorithm that, when given a
specified vertex $v$ and then the graph as a random order stream,
$\varepsilon$-approximates the distribution of length-$k$ random walks on
graphs drawn from the family uses $\bOm{n}$ bits of space, with a constant
factor depending only on $\varepsilon$.
\end{theorem}
The proof will be a simplified version of the hardness proof for digraphs given
in Section~\ref{sec:dglb}. We will prove that any algorithm giving such an
approximation gives a protocol solving $\ind_{n-3}$ on a random
instance with probability a constant factor greater than $1/4$. We restate the
lower bound on $\ind_n$ here for convenience. Recall that for $\ind_n$ Alice's
input is a string $x \in \bool^n$, Bob's an index $I \in \brac{n}$, and their
objective is for Alice to send Bob a message such that Bob can determine $x_I$.

\indistlb*

\subsection{Graph Distribution}
\label{sec:hardchosenvertdist}
We start by defining a distribution on length-$(n-3)$ undirected graph streams
that Alice and Bob can construct (with a prefix belonging to Alice and a
postfix belonging to Bob) using their inputs to $\ind_{n-3}$.

We will show that the graph streams correspond to randomly choosing a graph and
then uniformly permuting its edges, and any algorithm that generates a
distribution that is $\varepsilon$-close to the distribution of walks starting
at a specified vertex for some constant $\varepsilon < 1/4 -
2^{-2-\floor{\frac{k}{2}}}$ will be able to use this to solve $\ind_{n-3}$.

\paragraph{Vertices.} There will be $1$ vertex $a$, $n-3$ vertices
$(b_i)_{i=1}^{n-3}$, and two vertices $c_0,
c_1$. $a$ will be in a $2$-edge component with one of $c_0, c_1$.

\paragraph{Alice's edges.} Let $\pi$ be a uniformly random permutation of
$\brac{n}$, and let $J$ be drawn uniformly from $\set{0, \dots, n}$.  These
will be used to define the boundary between Bob's edges and Alice's edges.

Recall that Alice's input is a string $x \in \bool^n$. For each $i \in
\brac{J}$, Alice has the edge $b_{\pi(i)}c_{x_{\pi(i)}}$, with
$b_{\pi(1)}c_{x_{\pi(1)}}$ first, $b_{\pi(2)}c_{x_{\pi(2)}}$ second, and so on.

\paragraph{Bob's edges} Recall that Bob has the index $I$. His first edge will
be $ab_I$. Then, for each $i \in \set{J+1, \dots n}$, he has the edge
$b_{\pi(i)}c_{y_{\pi(i)}}$, where $y$ is a random $n$-bit string.

\begin{lemma}
\label{lm:chosevertexisrand}
The graph stream described above is a uniformly random order graph stream.
\end{lemma}
\begin{proof}
Fix any value of Bob's index $I$. If we were to randomly draw a string $z \in
\bool^n$, randomly order edges $(b_ic_{z_i})_{i=1}^n$, and then insert $ab_I$
randomly in this stream, this would give a uniformly random order graph stream
(corresponding to drawing a graph from a fixed distribution and then randomly
permuting its edges).

But this would also be identically distributed to our graph stream, as the
strings $x$ and $y$ are both drawn uniformly at random from $\bool^n$. So we
have a uniformly random graph stream.
\end{proof}

\begin{lemma}
\label{lm:chosenvertexcorr}
With probability $3/4$, $a$ at one end of a $2$-edge path with
$c_{x_I}$ at the other end. Otherwise it is at one end of a $2$-edge path with
$c_{\overline{x_I}}$ at the other end.
\end{lemma}
\begin{proof}
There is an edge between $a$ and $b_J$, and with probability $1/2$, $\pi(I) \le
J$, and so the only other edge incident to $b_J$ is incident to $c_{x_I}$.
Otherwise, it is incident to $c_{y_I}$, which is $c_{x_I}$ with probability
$1/2$, as $y$ is drawn at random.
\end{proof}

We will now prove that the random walk distribution from $a$ can be used to
determine the answer to $\ind_{n-3}$. 
\begin{lemma}
\label{lm:rwcvprot}
Let $\Ac$ be any algorithm using $S$ space such that, when given a sample from
the distribution on graph streams above, the output of $\Ac$ is
$\varepsilon$-close in total variation distance to sampling a $k$-step random
walk from $a$, where $k \ge 2$.
Then there is a protocol for $\ind_{n-3}$ on uniform inputs that uses $S$ space and
succeeds with probability $\frac{3}{4} - 2^{-2-\floor{\frac{k}{2}}} -
\varepsilon$. 
\end{lemma}
\begin{proof}
The protocol will be as follows:
\begin{enumerate}
\item Alice and Bob use their $\ind$ input to construct a corresponding random
graph stream.
\item Using $S$ bits of communication from Alice to Bob, they run $\Ac$ on the
stream.
\item If the walk output by $\Ac$ ever reaches $c_b$ for some $b \in \bool$,
output $i$. Otherwise output a random $b \in \bool$.
\end{enumerate}
Now, by Lemma~\ref{lm:chosenvertexcorr}, $a$ is at one end of a $2$-edge path
with one of $c_0, c_1$ at the other end, and there is a $3/4$ chance it is
$c_{x_I}$. So there is a \[
1 - 2^{-\floor{\frac{k}{2}}}
\]
probability any walk from $a$ reaches a vertex in $c_0, c_1$. If the walk
output by $\Ac$ does this the protocol succeeds with probability $3/4$ and
otherwise it succeeds with probability $1/2$. So the correct answer is returned
with probability at least
\[
(1 - 2^{-\floor{\frac{k}{2}}})\frac{3}{4} +
(2^{-\floor{\frac{k}{2}}})\frac{1}{2} - \varepsilon = \frac{3}{4} -
2^{-2-\floor{\frac{k}{2}}} - \varepsilon
\]
proving the correctness of the protocol. The construction of the stream itself
is done entirely with public randomness, so this gives an $S$ bit public
randomness protocol. As we are working with a fixed input distribution, this
means that there is also an $S$ bit private randomness protocol, by fixing
whichever set of public random bits maximizes the probability of success over
the uniform distribution.
\end{proof}

This then proves Theorem~\ref{thm:chosenvertexlb}, as any algorithm that can
output a constant $\varepsilon < 1/4 - 2^{-2-\floor{\frac{k}{2}}}$ approximation to
the random walk distribution will give a protocol for $\ind_{n-3}$ that works
with a probability at least a constant greater than $1/2$.